\documentclass{article}
\usepackage{amssymb}
\usepackage{amsfonts}
\usepackage{amsmath}

\newtheorem{theorem}{Theorem}

\newtheorem{corollary}[theorem]{Corollary}

\newtheorem{definition}[theorem]{Definition}
\newtheorem{example}[theorem]{Example}

\newtheorem{proposition}[theorem]{Proposition}
\newtheorem{remark}[theorem]{Remark}

\newenvironment{proof}[1][Proof]{\noindent\textbf{#1.} }{\ \rule{0.5em}{0.5em}}
\input{tcilatex}
\begin{document}

\title{Pseudotensor Problem of Gravitational Energy-momentum and Noether's
Theorem Revisited}
\author{Zhaoyan Wu \\
Center for Theoretical Physics, Jilin University, China}
\maketitle

\begin{abstract}
Based on a general variational principle, Noether's theorem is revisited. It
is shown that the long existing pseudotensor problem of gravitational
energy-momentum is only a result of misreading Noether's theorem and
mistaking different geometrical, physical objects as one and the same. As a
matter of fact, all the Noether's conserved quantities in general relativity
are scalars, their conservation currents are vector fields on spacetime
manifold. The difficulty of non-localizability of gravitational
energy-momentum, which is a direct consequence of the pseudotensor property,
does not really exist.
\end{abstract}

\section{Introduction}

Einstein had no sooner founded his theory of general relativity (GR) than he
realized its equation of motion does not give a continuity equation for the
conservation of energy-momentum of matter $T^{\alpha \beta }$. In order to
keep the law of conservation of energy-momentum alive in GR, Einstein
introduced the gravitational energy-momentum $t^{\alpha \beta }$ so that the
total energy-momentum $T^{\alpha \beta }+t^{\alpha \beta }$ is conserved.
Bauer immediately pointed out that this $t^{\alpha \beta }$ is not a tensor,
hence it is not localizable. Besides, it is asymmetrical. Non-symmetrical
stress is not allowed in physics. Decades later, Landau and Lifshitz found a
symmetrical gravitational energy-momentum[1], which is, however, still not a
tensor. Several gravitational energy-momentum complexes have been proposed.
They are all pseudotensors. It is commonly believed that no gravitational
energy-momentum complex can be a tensor.\footnote{%
In Appendix A, a proof of the following proposition is presented: If $%
S^{\alpha \beta }(x)$ is a tensor, and $\partial _{\alpha }[\sqrt{-|g(x)|}%
S^{\alpha \beta }(x)]=0$, for all coordinate systems, then $S^{\alpha \beta
}(x)=-S^{\beta \alpha }(x)$} A direct consequence of the pseudotensor
property of $t^{\alpha \beta }$ is the non-localizability of the
gravitational energy. It is not allowed to talk about the density of the
gravitational energy at a given spacetime point. It is not allowed either to
talk about the amount of gravitational energy on a finite space-like
hyper-surface bounded with a topological 2-sphere. But even so, many
relativists don't think the pseudotensor property and\ non-localizability of
gravitational energy have messed up GR, the most beautiful theory in
physics. They attribute the\ non-localizability of gravitational energy to
the equivalence principle physically, and to the following fact\
mathematically: For any geodesic $G$ in spacetime, one can always choose
coordinates such that the Chritoffel symbols $\Gamma _{\beta \gamma
}^{\alpha }$ at all $p\in G$ vanish. Non-localizability of gravitational
energy is considered inherent in the theory of GR.

In the present paper, based on a general variational principle of classical
fields, Noether's theorem is revisited. It is shown that in GR,
corresponding to the 1-parameter local groups of coordinate transformations
with the same form (for example, $x^{\alpha}\mapsto\widetilde{x}%
^{\alpha}=x^{\alpha }+\epsilon\delta_{0}^{\alpha}$, $y^{\alpha}\mapsto%
\widetilde{y}^{\alpha }=y^{\alpha}+\epsilon\delta_{0}^{\alpha}$, $%
z^{\alpha}\mapsto\widetilde {z}^{\alpha}=z^{\alpha}+\epsilon\delta_{0}^{%
\alpha}$, etc.), the conserved expressions in different coordinate systems
are functions of coordinates, field quantities and their derivatives with
the same form, but they are not the components in different coordinate
systems of the same geometrical, physical object; while corresponding to the
same 1-parameter local group of diffeomorphisms of spacetime $M$ onto
itself, the conserved expressions in different coordinate systems are
functions of coordinates, field quantities and their derivatives with
(generally) different forms, but they are the components of the same
geometrical, physical object. All the Noether's conserved quantities in GR
are scalars (the corresponding Noether's conservation currents are vector
fields on the spacetime manifold); the long existing pseudotensor problem of
gravitational energy-momentum is only a result of misreading Noether's
theorem and mistaking different geometrical, physical objects as one and the
same; and the non-localizability difficulty does not really exist.

In section 2, a general variational principle for classical fields is
presented. In section 3, the Noether's theorem is rederived. Then, in
section 4, these general results are applied to the specific case of general
relativity, especially Noether's theorem is applied to get quite a few
conservation laws in GR. It is noted that the conserved quantities obtained
here are not tensors, like various gravitational energy-momentum complexes.
Then Noether's theorem is revisited in section 5, and there the main
theorems are proved. In section 6, physical significance of the main
theorems are discussed.

These results will be used to explore the energy-momentum conservation and
the gravitational energy-momentum in GR in a later paper.

\section{Variational principle for classical fields}

There have been varied versions of variational principle and Noether's
theorem in the literature[2,3], and different notations have been used by
different authors. For the readers' convenience, and for the consistency of
the reasoning, we start with presenting a general variational principle for
classical fields in $n(\geqslant 2)$-dimensional spacetime with a Lagrangian
containing the spacetime coordinates, field quantities and their derivatives
of up to the $N(\geqslant 1)$-th order. In our opinion, non-local
interaction is not acceptable, so we assume that the Lagrangian does not
contain integrations of the field quantities.

First, we present a useful mathematical formula for the variational
principle, which does not rely on physics. Suppose $\left\{ \Phi_{B}:\mathbb{%
R}^{n}\longrightarrow\mathbb{R}|B=1,2,\ldots,f\right\} $ $(n\geqslant2)$ are
smooth functions, and function $L=L(x,\Phi(x),\partial\Phi(x),\ldots
,\partial^{N}\Phi(x))$ is smooth with respect to all its arguments. It is
easy to show just by using Leibniz's rule that (the Einstein convention is
used for coordinate indices)\footnote{%
Some arguments of $L$ are not independent of each other, such as $\frac{%
\partial^{2}}{\partial x^{1}\partial x^{2}}\Phi _{B}(x)=\frac{\partial^{2}}{%
\partial x^{2}\partial x^{1}}\Phi_{B}(x),$ $g_{\alpha\beta}(x)=g_{\beta%
\alpha}(x)$, etc. In order to avoid the indefiniteness of derivatives
related to the above mentioned facts, and in order to keep the formulaes
neat, it is assumed in the present paper without loss of generality that 
\begin{equation*}
\frac{\partial}{\partial\partial_{1}\partial_{2}\Phi_{B}}L=\frac{\partial }{%
\partial\partial_{2}\partial_{1}\Phi_{B}}L,\text{ \ }\frac{\partial }{%
\partial g_{\alpha\beta}(x)}L=\frac{\partial}{\partial g_{\beta\alpha}(x)}%
L,\ldots
\end{equation*}
See Appendix B for the details.}%
\begin{equation*}
\delta L=\sum_{B=1}^{f}\sum_{X=0}^{N}\overline{\delta}\partial_{\lambda_{1}}%
\cdots\partial_{\lambda_{X}}\Phi_{B}(x)\frac{\partial L}{\partial
\partial_{\lambda_{1}}\cdots\partial_{\lambda_{X}}\Phi_{B}(x)}
\end{equation*}%
\begin{equation*}
=\sum_{B=1}^{f}\overline{\delta}\Phi_{B}(x)\sum_{X=0}^{N}(-1)^{X}\partial_{%
\lambda_{1}}\cdots\partial_{\lambda_{X}}\frac{\partial L}{%
\partial\partial_{\lambda_{1}}\cdots\partial_{\lambda_{X}}\Phi_{B}(x)}%
+\partial_{\lambda}\sum_{B=1}^{f}\sum_{Y=0}^{N-1}
\end{equation*}%
\begin{equation}
\overline{\delta}\partial_{\mu_{1}}\cdots\partial_{\mu_{Y}}\Phi_{B}(x)%
\sum_{Z=0}^{N-1-Y}(-1)^{Z}\partial_{\nu_{1}}\cdots\partial_{\nu_{Z}}\frac{%
\partial L}{\partial\partial_{\lambda}\partial_{\mu_{1}}\cdots
\partial_{\mu_{Y}}\partial_{\nu_{1}}\cdots\partial_{\nu_{Z}}\Phi_{B}(x)} 
\tag{1}
\end{equation}
when $\Phi_{B}(x)\longmapsto\widetilde{\Phi}_{B}(x)=\Phi_{B}(x)+\overline {%
\delta}\Phi_{B}(x),\forall B=1,\ldots,f$. Consider the functional $F$ of the
following form%
\begin{equation}
F[\Phi]=\int_{\Omega}d^{n}xL(x,\Phi(x),\partial\Phi(x),\partial^{2}\Phi(x),%
\ldots,\partial^{N}\Phi(x))\text{,}  \tag{2}
\end{equation}
where $\Omega$ is an open subset with a compact closure of $\mathbb{R}^{n}$.
When the arguments $\Phi_{B}(x)$ $\forall B=1,\ldots,f$ change slightly, the
variation of functional $F$ is 
\begin{equation*}
\delta F[\Phi]=\int_{\Omega}d^{n}x\sum_{B=1}^{f}\sum_{X=0}^{N}\overline {%
\delta}\partial_{\lambda_{1}}\cdots\partial_{\lambda_{X}}\Phi_{B}(x)\frac{%
\partial L}{\partial\partial_{\lambda_{1}}\cdots\partial_{\lambda
_{X}}\Phi_{B}(x)}
\end{equation*}%
\begin{align*}
& =\int_{\Omega}d^{n}x\sum_{B=1}^{f}\overline{\delta}\Phi_{B}(x)\sum
_{X=0}^{N}(-1)^{X}\partial_{\lambda_{1}}\cdots\partial_{\lambda_{X}}\frac{%
\partial L}{\partial\partial_{\lambda_{1}}\cdots\partial_{\lambda_{X}}%
\Phi_{B}(x)} \\
& +\int_{\Omega}d^{n}x\partial_{\lambda}\sum_{B=1}^{f}\sum_{Y=0}^{N-1}%
\overline{\delta}\partial_{\mu_{1}}\cdots\partial_{\mu_{Y}}\Phi_{B}(x) \\
& \sum_{Z=0}^{N-1-Y}(-1)^{Z}\partial_{\nu_{1}}\cdots\partial_{\nu_{Z}}\frac{%
\partial L}{\partial\partial_{\lambda}\partial_{\mu_{1}}\cdots
\partial_{\mu_{Y}}\partial_{\nu_{1}}\cdots\partial_{\nu_{Z}}\Phi_{B}(x)}
\end{align*}%
\begin{align}
& =\int_{\Omega}d^{n}x\sum_{B=1}^{f}\overline{\delta}\Phi_{B}(x)\sum
_{X=0}^{N}(-1)^{X}\partial_{\lambda_{1}}\cdots\partial_{\lambda_{X}}\frac{%
\partial L}{\partial\partial_{\lambda_{1}}\cdots\partial_{\lambda_{X}}%
\Phi_{B}(x)}  \notag \\
& +\int_{\partial\Omega}ds_{\lambda}(x)\sum_{B=1}^{f}\sum_{Y=0}^{N-1}%
\overline{\delta}\partial_{\mu_{1}}\cdots\partial_{\mu_{Y}}\Phi _{B}(x) 
\notag \\
& \sum_{Z=0}^{N-1-Y}(-1)^{Z}\partial_{\nu_{1}}\cdots\partial_{\nu_{Z}}\frac{%
\partial L}{\partial\partial_{\lambda}\partial_{\mu_{1}}\cdots
\partial_{\mu_{Y}}\partial_{\nu_{1}}\cdots\partial_{\nu_{Z}}\Phi_{B}(x)} 
\tag{3}
\end{align}
This can be easily obtained by using eqn.(1) and the Stokes theorem. The
derivatives of functional (2) is defined as follows.

\begin{definition}
If the change of the functional (2) can be expressed as%
\begin{equation}
F[\Phi+\overline{\delta}\Phi]-F[\Phi]=\int_{\Omega}d^{n}x\sum_{B=1}^{f}%
\overline{\delta}\Phi_{B}(x)D^{B}[\Phi,x]+o[\overline{\delta}\Phi]  \tag{4}
\end{equation}
where $D^{B}[\Phi,x]$ is a functional of $\Phi$ varying with $x$, and $o[%
\overline{\delta}\Phi]$ is a higher order infinitesimal of $\overline {\delta%
}\Phi$, when $\left\{ \Phi_{B}\text{ }|\text{ }B=1,2,\ldots,f\right\} $
change slightly while the boundary values of $\Phi,\partial\Phi
,\ldots,\partial^{N-1}\Phi$ are kept fixed, then $F$ is called
differentialble at $\Phi$, and $D^{B}[\Phi,x]$ is called the derivative of
functional $F$ with respect to $\Phi_{B}$ at $\Phi$ and point $x$, and
denoted by
\end{definition}

\begin{equation}
D^{B}[\Phi,x]=\frac{\delta F[\Phi]}{\delta\Phi_{B}(x)}  \tag{5}
\end{equation}

Let us now apply the general formula (3) to the action functional of
classical field $\left\{ \Phi_{B}:M\longrightarrow\mathbb{R}\text{ }|\text{ }%
B=1,2,\ldots,f\right\} $

\begin{equation}
A[\Phi]=\int_{x(\Omega)}d^{n}xL(x,\Phi(x),\partial\Phi(x),\partial^{2}%
\Phi(x),\ldots,\partial^{N}\Phi(x))  \tag{6}
\end{equation}
where $\Omega$ is an open subset with a compact closure of the spacetime
manifold $M$, $x(\Omega)\subset\mathbb{R}^{n}$ is the image of $%
\Omega\subset M$ under the coordinate mapping $x:M\mathbb{\rightarrow R}^{n}$
and $L$ is the Lagrangian of the field. We get the difference between the
action functionals over $\Omega$ of two kinematically allowed movements
close to each other

\begin{equation*}
\delta A[\Phi]=\int_{x(\Omega)}d^{n}x\sum_{B=1}^{f}\overline{\delta}\Phi
_{B}(x)\sum_{X=0}^{N}(-1)^{X}\partial_{\lambda_{1}}\cdots\partial_{\lambda
_{X}}\frac{\partial L}{\partial\partial_{\lambda_{1}}\cdots\partial
_{\lambda_{X}}\Phi_{B}(x)}
\end{equation*}%
\begin{align*}
& +\int_{x(\Omega)}d^{n}x\partial_{\lambda}\sum_{B=1}^{f}\sum_{Y=0}^{N-1}%
\overline{\delta}\partial_{\mu_{1}}\cdots\partial_{\mu_{Y}}\Phi_{B}(x) \\
& \sum_{Z=0}^{N-1-Y}(-1)^{Z}\partial_{\nu_{1}}\cdots\partial_{\nu_{Z}}\frac{%
\partial L}{\partial\partial_{\lambda}\partial_{\mu_{1}}\cdots
\partial_{\mu_{Y}}\partial_{\nu_{1}}\cdots\partial_{\nu_{Z}}\Phi_{B}(x)}
\end{align*}%
\begin{equation*}
=\int_{x(\Omega)}d^{n}x\sum_{B=1}^{f}\overline{\delta}\Phi_{B}(x)\sum
_{X=0}^{N}(-1)^{X}\partial_{\lambda_{1}}\cdots\partial_{\lambda_{X}}\frac{%
\partial L}{\partial\partial_{\lambda_{1}}\cdots\partial_{\lambda_{X}}%
\Phi_{B}(x)}
\end{equation*}%
\begin{align}
& +\int_{x(\partial\Omega)}ds_{\lambda}(x)\sum_{B=1}^{f}\sum_{Y=0}^{N-1}%
\overline{\delta}\partial_{\mu_{1}}\cdots\partial_{\mu_{Y}}\Phi _{B}(x) 
\notag \\
& \sum_{Z=0}^{N-1-Y}(-1)^{Z}\partial_{\nu_{1}}\cdots\partial_{\nu_{Z}}\frac{%
\partial L}{\partial\partial_{\lambda}\partial_{\mu_{1}}\cdots
\partial_{\mu_{Y}}\partial_{\nu_{1}}\cdots\partial_{\nu_{Z}}\Phi_{B}(x)} 
\tag{7}
\end{align}

Equation\ (7) suggests that for all $N\geqslant1$, $n\geqslant2$ \textbf{the
least action principle} read as follows.

\textit{For any spacetime region }$\Omega $ with a compact closure, a\textit{%
mong all kinematically allowed movements in }$\Omega $ \textit{with the same
boundary condition}

\begin{equation}
\delta\Phi|_{\partial\Omega}=0,\delta\partial\Phi|_{\partial\Omega}=0,%
\ldots,\delta\partial^{N-1}\Phi|_{\partial\Omega}=0,  \tag{8}
\end{equation}
\textit{the movement allowed by physical laws takes the stationary value of
the action over }$\Omega$\textit{.}

Combining eqns.(7), (8), one obtains\textbf{\ the equation of motion}

\begin{equation}
\frac{\delta A[\Phi]}{\delta\Phi_{B}(x)}=\sum_{X=0}^{N}(-1)^{X}\partial
_{\lambda_{1}}\cdots\partial_{\lambda_{X}}\frac{\partial L}{\partial
\partial_{\lambda_{1}}\cdots\partial_{\lambda_{X}}\Phi_{B}(x)}=0.  \tag{9}
\end{equation}

\section{Noether's theorem}

\subsection{Re-deriving the theorem}

Performing symmetry analysis, one can adopt the active viewpoint, or the
passive viewpoint. They are equivalent to each other, but the former is more
elegant. What I am going to do, however, is to show how the pseudotensor
problem of gravitational energy-momentum arises, and it is related to the
passive viewpoint. So I will use the coordinate language in the following.

\begin{theorem}
If the action of classical fields over every spacetime region $\Omega$ with
a compact closure%
\begin{equation*}
A[\Phi]=\int_{x(\Omega)}d^{n}xL(x,\Phi(x),\partial\Phi(x),\partial^{2}%
\Phi(x),\ldots,\partial^{N}\Phi(x))
\end{equation*}
remains unchanged under the following $r-$parameter local group of
coordinate transformations 
\begin{align}
x^{\lambda} & \longmapsto\widetilde{x}^{\lambda}=x^{\lambda}+\delta
x^{\lambda},  \notag \\
\delta x^{\lambda} & =:\delta x^{\lambda}(x,\epsilon^{1},\ldots,\epsilon
^{r}),|\epsilon^{i}|\ll1,\delta x^{\lambda}(x,0,\ldots,0)=0\text{ \ \ \ \ } 
\notag \\
\Phi_{B}(x) & \longmapsto\widetilde{\Phi}_{B}(\widetilde{x})=\Phi
_{B}(x)+\delta\Phi_{B}(x),\text{ }  \notag \\
\delta\Phi_{B}(x) & =:\delta\Phi_{B}(x,\epsilon^{1},\ldots,\epsilon ^{r}),%
\text{ }|\epsilon^{i}|\ll1,\text{ }\delta\Phi_{B}(x,0,\ldots,0)=0  \tag{10}
\end{align}
then there exist $r$ conservation laws.
\end{theorem}

\begin{proof}
The small change of field quantity $\delta\Phi_{B}(x)$ can be devided into
two parts, the part due to the small change of its function form and the
part due to the small change of its coordinate arguments.

\begin{equation}
\overline{\delta }\Phi _{B}(x)=:\widetilde{\Phi }_{B}(x)-\Phi _{B}(x),\delta
\Phi _{B}(x)=\overline{\delta }\Phi _{B}(x)+\delta x^{\sigma }\partial
_{\sigma }\Phi _{B}(x)  \tag{11}
\end{equation}%
Similarily the small change of derivatives of field quantity $\delta \lbrack
\partial _{\lambda _{1}}\partial _{\lambda _{2}}\cdots \partial _{\lambda
_{X}}\Phi _{B}(x)]$ can be written as%
\begin{align}
\delta \lbrack \partial _{\lambda _{1}}\partial _{\lambda _{2}}\cdots
\partial _{\lambda _{X}}\Phi _{B}(x)]& =\partial _{\lambda _{1}}\partial
_{\lambda _{2}}\cdots \partial _{\lambda _{X}}\overline{\delta }\Phi _{B}(x)
\notag \\
& +\delta x^{\sigma }\partial _{\sigma }\partial _{\lambda _{1}}\partial
_{\lambda _{2}}\cdots \partial _{\lambda _{X}}\Phi _{B}(x)  \tag{12}
\end{align}%
The variation of the action can be\ devided into two parts. One is due to
the small change of the integration domain $x(\Omega )\longmapsto \widetilde{%
x}(\Omega )$ in $\mathbb{R}^{n}$, and the other is due to the small change
of the integrand 
\begin{align*}
& L(x,\Phi (x),\partial \Phi (x),\partial ^{2}\Phi (x),\ldots ,\partial
^{N}\Phi (x)) \\
& \longmapsto L(x,\widetilde{\Phi }(x),\partial \widetilde{\Phi }%
(x),\partial ^{2}\widetilde{\Phi }(x),\ldots ,\partial ^{N}\widetilde{\Phi }%
(x))
\end{align*}%
\begin{align}
\delta A[\Phi ]& =\int_{x(\partial \Omega )}ds_{\lambda }(x)\delta
x^{\lambda }L  \notag \\
& +\int_{x(\Omega )}d^{n}x\sum\limits_{B=1}^{f}\sum\limits_{X=0}^{N}\frac{%
\partial L}{\partial \partial _{\mu _{1}}\cdots \partial _{\mu _{X}}\Phi
_{B}(x)}\partial _{\mu _{1}}\cdots \partial _{\mu _{X}}\overline{\delta }%
\Phi _{B}(x)  \notag
\end{align}%
\begin{align}
& =\int_{x(\Omega )}d^{n}x\sum\limits_{B=1}^{f}\overline{\delta }\Phi
_{B}(x)\sum\limits_{X=0}^{N}(-1)^{X}\partial _{\mu _{1}}\cdots \partial
_{\mu X}\frac{\partial L}{\partial \partial _{\mu _{1}}\cdots \partial _{\mu
_{X}}\Phi _{B}(x)}  \notag \\
& +\int_{x(\Omega )}d^{n}x\partial _{\lambda }[\delta x^{\sigma }\delta
_{\sigma }^{\lambda }L+\sum\limits_{B=1}^{f}\sum\limits_{Y=0}^{N-1}\partial
_{\mu _{1}}\cdots \partial _{\mu _{Y}}\overline{\delta }\Phi _{B}(x)\times 
\notag
\end{align}%
\begin{equation*}
\sum_{Z=0}^{N-1-Y}(-1)^{Z}\partial _{\nu _{1}}\cdots \partial _{\nu _{Z}}%
\frac{\partial L}{\partial \partial _{\lambda }\partial _{\mu _{1}}\cdots
\partial _{\mu _{Y}}\partial _{\nu _{1}}\cdots \partial _{\nu _{Z}}\Phi
_{B}(x)}]=0
\end{equation*}%
Due to the arbitrariness of $\Omega $, we get the following identities,
which hold for all kinematically allowed movements. 
\begin{align}
& \sum\limits_{B=1}^{f}\overline{\delta }\Phi
_{B}(x)\sum\limits_{X=0}^{N}(-1)^{X}\partial _{\mu _{1}}\cdots \partial
_{\mu X}\frac{\partial L}{\partial \partial _{\mu _{1}}\cdots \partial _{\mu
_{X}}\Phi _{B}(x)}  \notag \\
& +\partial _{\lambda }[\delta x^{\sigma }\delta _{\sigma }^{\lambda
}L+\sum\limits_{B=1}^{f}\sum\limits_{Y=0}^{N-1}\partial _{\mu _{1}}\cdots
\partial _{\mu _{Y}}\overline{\delta }\Phi _{B}(x)\times  \notag
\end{align}%
\begin{equation}
\sum_{Z=0}^{N-1-Y}(-1)^{Z}\partial _{\nu _{1}}\cdots \partial _{\nu _{Z}}%
\frac{\partial L}{\partial \partial _{\lambda }\partial _{\mu _{1}}\cdots
\partial _{\mu _{Y}}\partial _{\nu _{1}}\cdots \partial _{\nu _{Z}}\Phi
_{B}(x)}]=0  \tag{13}
\end{equation}%
For movements allowed by physical laws, the first line of eqn.(13) vanishes,
hence we get the following continuity equation. 
\begin{equation*}
\partial _{\lambda }[\delta x^{\sigma }\delta _{\sigma }^{\lambda
}L+\sum\limits_{B=1}^{f}\sum\limits_{Y=0}^{N-1}\partial _{\mu _{1}}\cdots
\partial _{\mu _{Y}}\left( \delta \Phi _{B}(x)-\delta x^{\sigma }\partial
_{\sigma }\Phi _{B}(x)\right)
\end{equation*}%
\begin{equation}
\sum_{Z=0}^{N-1-Y}(-1)^{Z}\partial _{\nu _{1}}\cdots \partial _{\nu _{Z}}%
\frac{\partial L}{\partial \partial _{\lambda }\partial _{\mu _{1}}\cdots
\partial _{\mu _{Y}}\partial _{\nu _{1}}\cdots \partial _{\nu _{Z}}\Phi
_{B}(x)}]=0  \tag{14}
\end{equation}%
or%
\begin{equation*}
\int_{x(\partial \Omega )}ds_{\lambda }(x)\{\delta x^{\sigma }\delta
_{\sigma }^{\lambda }L+\sum\limits_{B=1}^{f}\sum\limits_{Y=0}^{N-1}[\partial
_{\mu _{1}}\cdots \partial _{\mu _{Y}}\left( \delta \Phi _{B}(x)-\delta
x^{\sigma }\partial _{\sigma }\Phi _{B}(x)\right)
\end{equation*}%
\begin{equation}
\sum_{Z=0}^{N-1-Y}(-1)^{Z}\partial _{\nu _{1}}\cdots \partial _{\nu _{Z}}%
\frac{\partial L}{\partial \partial _{\lambda }\partial _{\mu _{1}}\cdots
\partial _{\mu _{Y}}\partial _{\nu _{1}}\cdots \partial _{\nu _{Z}}\Phi
_{B}(x)}]\}=0  \tag{15}
\end{equation}%
Noting that both $\delta x^{\sigma }$ and $\delta \Phi _{B}(x)$ depend on
the parameters $\epsilon ^{\alpha }$, one gets from eqn.(14)$\ $the
following $r$ conservation laws.%
\begin{align}
& \partial _{\lambda }\{\frac{\partial \delta x^{\sigma }}{\partial \epsilon
^{\alpha }}|_{\epsilon =0}\delta _{\sigma }^{\lambda
}L+\sum\limits_{B=1}^{f}\sum\limits_{Y=0}^{N-1}\partial _{\mu _{1}}\cdots
\partial _{\mu _{Y}}\left( \frac{\partial \delta \Phi _{B}(x)}{\partial
\epsilon ^{\alpha }}|_{\epsilon =0}-\frac{\partial \delta x^{\sigma }}{%
\partial \epsilon ^{\alpha }}|_{\epsilon =0}\partial _{\sigma }\Phi
_{B}(x)\right)  \notag \\
& \sum_{Z=0}^{N-1-Y}(-1)^{Z}\partial _{\nu _{1}}\cdots \partial _{\nu _{Z}}%
\frac{\partial L}{\partial \partial _{\lambda }\partial _{\mu _{1}}\cdots
\partial _{\mu _{Y}}\partial _{\nu _{1}}\cdots \partial _{\nu _{Z}}\Phi
_{B}(x)}=0,\forall \alpha =1,\ldots ,r  \tag{16}
\end{align}%
\textit{\ }
\end{proof}

The formalism presented so far is good for any classical field in $n$%
-dimensional spacetime with a Lagrangian containing spacetime coordinates,
field quantities and their derivatives of up to the $N$-th order, no matter
the Lagrangian is Galileo covariant, Lorentz covariant, generally covariant
or without any covariance.

\section{Variational principle approach to general relativity}

Let us apply the general results obtained above to the classic fields in GR.
We will consider the case of $(1,1)$-tensor matter field. The results can be
readily generalized to any $(r,s)-$tensor matter field. For the dynamic
system, $(1,1)$-tensor matter field $u_{\xi}^{\theta}(x)$ plus the metric
field $g_{\alpha\beta}(x)$, the Lagrangian and the action over spacetime
region $\Omega$ are respectively

\begin{align}
& L(g(x),\partial g(x),\partial^{2}g(x),u(x),\partial u(x))  \notag \\
& =\sqrt{-|g(x)|}[\mathcal{L}(g(x),u(x),\nabla u(x))+\frac{1}{16\pi G}%
R]=L_{M}+L_{G},  \tag{17}
\end{align}%
\begin{equation*}
A[g,u]=\int_{x(\Omega)}d^{4}x\sqrt{-|g(x)|}[\mathcal{L}(g(x),u(x),\nabla
u(x))
\end{equation*}%
\begin{equation}
+\frac{1}{16\pi G}R]=A_{M}[g,u]+A_{G}[g]  \tag{18}
\end{equation}
where $\mathcal{L}(g(x),u(x),\nabla u(x))$ is the sum of a few scalars
obtained by contracting $g_{\alpha\beta}(x)$, $u_{\varphi}^{\theta}(x)$ and $%
\nabla_{\lambda}u_{\varphi}^{\theta}(x)$ and multiplying the contractions
with suitable numbers such that $\mathcal{L}(\eta,u(x),\partial u(x))$ is
the Lagrangian in special relativity ($\eta=diag(-1,1,1,1)$).

\subsection{Einstein's field equation}

The Euler-Lagrange equation, Eqn.(9) now reads

\begin{equation}
\frac{\delta A[g,u]}{\delta u_{\varphi}^{\theta}(x)}=\sqrt{-|g(x)|}[\frac{%
\partial\mathcal{L}}{\partial u_{\varphi}^{\theta}(x)}-\nabla_{\lambda }%
\frac{\partial\mathcal{L}}{\partial\nabla_{\lambda}u_{\varphi}^{\theta}(x)}%
]=0  \tag{19}
\end{equation}

\begin{equation}
\frac{\delta A[g,u]}{\delta g_{\alpha\beta}(x)}=\sqrt{-|g(x)|}\frac{1}{16\pi
G}[R^{\alpha\beta}-\frac{1}{2}Rg^{\alpha\beta}(x)-8\pi GT^{\alpha\beta}]=0 
\tag{20}
\end{equation}
where $T^{\alpha\beta}$ is the energy-momentum tensor of matter field, which
is a symmetrical (2,0)-tensor.

\begin{equation}
T^{\alpha\beta}=:\frac{2}{\sqrt{-|g(x)|}}\frac{\delta A_{M}[g,u]}{\delta
g_{\alpha\beta}(x)}=T^{\alpha\beta}(u(x),\partial u(x),g(x),\partial g(x)) 
\tag{21}
\end{equation}

\subsection{Noether's theorem for classical field in GR}

Action (18) is invariant under arbitrary coordinate transformations. The
Noether's conservation law, or the continuity equation (14) , now reads%
\begin{equation*}
\frac{\partial }{\partial x^{\kappa }}\{\sqrt{-|g(x)|}J^{\kappa
}[u(x),\partial u(x),g(x),\partial g(x),\partial ^{2}g(x);
\end{equation*}%
\begin{equation}
\delta x,\overline{\delta }u(x),\overline{\delta }g(x),\partial \overline{%
\delta }g(x)]\}=0  \tag{22}
\end{equation}%
where

\begin{align}
& J^{\kappa }\left[ u(x),\partial u(x),g(x),\partial g(x),\partial
^{2}g(x);\delta x,\overline{\delta }u(x),\overline{\delta }g(x),\partial 
\overline{\delta }g(x)\right] =:J_{x}^{\kappa }  \notag \\
& =\{\mathcal{L}\delta _{\rho }^{\kappa }\delta x^{\rho }+\frac{\partial 
\mathcal{L}}{\partial \nabla _{\kappa }u_{\xi }^{\theta }(x)}\overline{%
\delta }u_{\xi }^{\theta }(x)+\frac{1}{2}[\frac{\partial \mathcal{L}}{%
\partial \nabla _{\kappa }u_{\xi }^{\theta }(x)}g^{\theta \alpha }(x)u_{\xi
}^{\beta }(x)  \notag \\
& +\frac{\partial \mathcal{L}}{\partial \nabla _{\beta }u_{\xi }^{\theta }(x)%
}g^{\theta \alpha }(x)u_{\xi }^{\kappa }(x)-\frac{\partial \mathcal{L}}{%
\partial \nabla _{\alpha }u_{\xi }^{\theta }(x)}g^{\theta \kappa }(x)u_{\xi
}^{\beta }(x)-\frac{\partial \mathcal{L}}{\partial \nabla _{\kappa }u_{\beta
}^{\theta }(x)}g^{\xi \alpha }(x)u_{\xi }^{\theta }(x)  \notag
\end{align}%
\begin{align}
& -\frac{\partial \mathcal{L}}{\partial \nabla _{\beta }u_{\kappa }^{\theta
}(x)}g^{\xi \alpha }(x)u_{\xi }^{\theta }(x)+\frac{\partial \mathcal{L}}{%
\partial \nabla _{\alpha }u_{\beta }^{\theta }(x)}g^{\xi \kappa }(x)u_{\xi
}^{\theta }(x)]\overline{\delta }g_{\alpha \beta }(x)\}+  \notag \\
& \frac{1}{16\pi G}\{R\delta _{\rho }^{\kappa }\delta x^{\rho }+[\frac{%
\partial R}{\partial \partial _{\kappa }g_{\alpha \beta }(x)}-\partial _{\mu
}\frac{\partial R}{\partial \partial _{\kappa }\partial _{\mu }g_{\alpha
\beta }(x)}  \notag \\
& -\Gamma _{\nu \mu }^{\nu }(x)\frac{\partial R}{\partial \partial _{\kappa
}\partial _{\mu }g_{\alpha \beta }(x)}]\overline{\delta }g_{\alpha \beta
}(x)+\frac{\partial R}{\partial \partial _{\kappa }\partial _{\mu }g_{\alpha
\beta }(x)}\partial _{\mu }\overline{\delta }g_{\alpha \beta }(x)\}  \tag{23}
\end{align}%
It is worth noting that the form of function $J^{\kappa }$ is independent of
coordinate systems, and the arguments of function $J^{\kappa }$ are not only 
$u(x)$, $\partial u(x)$, $g(x)$, $\partial g(x)$, $\partial ^{2}g(x)$; but
also $\delta x$, $\overline{\delta }u(x)$, $\overline{\delta }g(x)$, $%
\partial \overline{\delta }g(x)$.

\subsubsection{Conservation law due to \textquotedblleft coordinate
shift\textquotedblright\ invariance}

Action (18) remains unchanged under the following \textquotedblleft
coordinate shifts\textquotedblright.

\begin{equation}
\delta x^{\rho}=\epsilon^{\rho},\text{ }\delta u_{\xi}^{\theta}(x)=0,\text{ }%
\delta g_{\alpha\beta}(x)=0.  \tag{24}
\end{equation}
In this case, eqn.(16) reads

\begin{equation}
\partial_{\lambda}[\sqrt{-|g(x)|}\tau_{\rho}^{\lambda}(x)]=0,  \tag{25}
\end{equation}
where

\begin{equation*}
\tau _{\rho }^{\kappa }(x)=:\tau _{\rho }^{\kappa }(u(x),\partial
u(x),g(x),\partial g(x),\partial ^{2}g(x))
\end{equation*}%
\begin{align}
& =\frac{\partial \mathcal{L}}{\partial \nabla _{\kappa }u_{\xi }^{\theta
}(x)}\partial _{\rho }u_{\xi }^{\theta }(x)-\mathcal{L}\delta _{\rho
}^{\kappa }+\frac{1}{2}[\frac{\partial \mathcal{L}}{\partial \nabla _{\kappa
}u_{\xi }^{\theta }(x)}g^{\theta \alpha }(x)u_{\xi }^{\beta }(x)  \notag \\
& +\frac{\partial \mathcal{L}}{\partial \nabla _{\beta }u_{\xi }^{\theta }(x)%
}g^{\theta \alpha }(x)u_{\xi }^{\kappa }(x)-\frac{\partial \mathcal{L}}{%
\partial \nabla _{\alpha }u_{\xi }^{\theta }(x)}g^{\theta \kappa }(x)u_{\xi
}^{\beta }(x)  \notag \\
& -\frac{\partial \mathcal{L}}{\partial \nabla _{\kappa }u_{\beta }^{\theta
}(x)}g^{\xi \alpha }(x)u_{\xi }^{\theta }(x)-\frac{\partial \mathcal{L}}{%
\partial \nabla _{\beta }u_{\kappa }^{\theta }(x)}g^{\xi \alpha }(x)u_{\xi
}^{\theta }(x)  \notag
\end{align}%
\begin{align}
& +\frac{\partial \mathcal{L}}{\partial \nabla _{\alpha }u_{\beta }^{\theta
}(x)}g^{\xi \kappa }(x)u_{\xi }^{\theta }(x)]\partial _{\rho }g_{\alpha
\beta }(x)+\frac{1}{16\pi G}  \notag \\
& [(\frac{\partial R}{\partial \partial _{\kappa }g_{\alpha \beta }(x)}%
-\partial _{\mu }\frac{\partial R}{\partial \partial _{\kappa }\partial
_{\mu }g_{\alpha \beta }(x)}-\frac{1}{2}g^{\eta \xi }(x)\partial _{\mu
}g_{\xi \eta }(x)  \notag \\
& \frac{\partial R}{\partial \partial _{\kappa }\partial _{\mu }g_{\alpha
\beta }(x)})\partial _{\rho }g_{\alpha \beta }(x)+\frac{\partial R}{\partial
\partial _{\kappa }\partial _{\mu }g_{\alpha \beta }(x)}\partial _{\mu
}\partial _{\rho }g_{\alpha \beta }(x)-R\delta _{\rho }^{\kappa }]  \tag{ 26}
\end{align}%
is usually called canonical energy-momentum. It is worth noting that the
arguments of function $\tau _{\rho }^{\kappa }(u(x),\partial
u(x),g(x),\partial g(x),\partial ^{2}g(x))$ are all field quantities and
their derivatives.

\subsubsection{Conservation law due to \textquotedblleft4-dimensional
rotation\textquotedblright\ invariance}

The action (18) remains unchanged under infinitesimal \textquotedblleft
4-dimensional rotations\textquotedblright\ (Lorentz transformations), which
form a 6-parameter family of infinitesimal symmetry transformations%
\begin{align}
x^{\mu }& \longmapsto \widetilde{x}^{\mu }=L_{\nu }^{\mu }x^{\nu },\text{ }%
L_{\nu }^{\mu }=\delta _{\nu }^{\mu }+\Lambda _{\nu }^{\mu },\text{ }%
|\Lambda _{\nu }^{\mu }|\ll 1,\text{ }\eta _{\mu \lambda }\Lambda _{\nu
}^{\lambda }\equiv \Lambda _{\mu \nu },\text{ }\Lambda _{\mu \nu }=-\Lambda
_{\nu \mu },  \notag \\
\delta x^{\lambda }& =\Lambda _{\mu }^{\lambda }x^{\mu }=\frac{1}{2}(\eta
^{\lambda \rho }x^{\sigma }-\eta ^{\lambda \sigma }x^{\rho })\Lambda _{\rho
\sigma },  \tag{27}
\end{align}%
\begin{equation*}
\delta u_{\xi }^{\theta }(x)=\Lambda _{\varphi }^{\theta }u_{\xi }^{\varphi
}(x)-\Lambda _{\xi }^{\eta }u_{\eta }^{\theta }(x)
\end{equation*}%
\begin{equation}
=\Lambda _{\rho \sigma }\frac{1}{2}[\eta ^{\theta \rho }u_{\xi }^{\sigma
}(x)-\eta ^{\theta \sigma }u_{\xi }^{\rho }(x)-\delta _{\xi }^{\sigma }\eta
^{\eta \rho }u_{\eta }^{\theta }(x)+\delta _{\xi }^{\rho }\eta ^{\eta \sigma
}u_{\eta }^{\theta }(x)],  \tag{28}
\end{equation}%
\begin{equation*}
\delta g_{\alpha \beta }(x)=\Lambda _{\rho \sigma }\frac{1}{2}[-\delta
_{\alpha }^{\sigma }\eta ^{\mu \rho }g_{\mu \beta }(x)
\end{equation*}%
\begin{equation}
+\delta _{\alpha }^{\rho }\eta ^{\mu \sigma }g_{\mu \beta }(x)-\delta
_{\beta }^{\sigma }\eta ^{\nu \rho }g_{\alpha \nu }(x)+\delta _{\beta
}^{\rho }\eta ^{\nu \sigma }g_{\alpha \nu }(x)]  \tag{29}
\end{equation}%
In this case eqn.(16) reads%
\begin{equation}
\frac{\partial }{\partial x^{\kappa }}[\sqrt{-|g(x)|}M^{\kappa \rho \sigma
}(x)]=0  \tag{30}
\end{equation}%
where%
\begin{align}
& 2M^{\kappa \rho \sigma }(u(x),\partial u(x),g(x),\partial g(x),\partial
^{2}g(x))=:2M^{\kappa \rho \sigma }(x)  \notag \\
& =\mathcal{L}\delta _{\lambda }^{\kappa }(\eta ^{\lambda \rho }x^{\sigma
}-\eta ^{\lambda \sigma }x^{\rho })+\frac{\partial \mathcal{L}}{\partial
\nabla _{\kappa }u_{\xi }^{\theta }(x)}\{[\eta ^{\theta \rho }u_{\xi
}^{\sigma }(x)-\eta ^{\theta \sigma }u_{\xi }^{\rho }(x)  \notag \\
& -\delta _{\xi }^{\sigma }\eta ^{\eta \rho }u_{\eta }^{\theta }(x)+\delta
_{\xi }^{\rho }\eta ^{\eta \sigma }u_{\eta }^{\theta }(x)]-(\eta ^{\lambda
\rho }x^{\sigma }-\eta ^{\lambda \sigma }x^{\rho })\partial _{\lambda
}u_{\xi }^{\theta }(x)\}  \notag \\
& +\frac{1}{2}[\frac{\partial \mathcal{L}}{\partial \nabla _{\kappa }u_{\xi
}^{\theta }(x)}g^{\theta \alpha }(x)u_{\xi }^{\beta }(x)+\frac{\partial 
\mathcal{L}}{\partial \nabla _{\beta }u_{\xi }^{\theta }(x)}g^{\theta \alpha
}(x)u_{\xi }^{\kappa }(x)  \notag
\end{align}%
\begin{align}
& -\frac{\partial \mathcal{L}}{\partial \nabla _{\alpha }u_{\xi }^{\theta
}(x)}g^{\theta \kappa }(x)u_{\xi }^{\beta }(x)-\frac{\partial \mathcal{L}}{%
\partial \nabla _{\kappa }u_{\beta }^{\theta }(x)}g^{\xi \alpha }(x)u_{\xi
}^{\theta }(x)  \notag \\
& -\frac{\partial \mathcal{L}}{\partial \nabla _{\beta }u_{\kappa }^{\theta
}(x)}g^{\xi \alpha }(x)u_{\xi }^{\theta }(x)+\frac{\partial \mathcal{L}}{%
\partial \nabla _{\alpha }u_{\beta }^{\theta }(x)}g^{\xi \kappa }(x)u_{\xi
}^{\theta }(x)]\times  \notag \\
& \{[-\delta _{\alpha }^{\sigma }\eta ^{\mu \rho }g_{\mu \beta }(x)+\delta
_{\alpha }^{\rho }\eta ^{\mu \sigma }g_{\mu \beta }(x)-\delta _{\beta
}^{\sigma }\eta ^{\nu \rho }g_{\alpha \nu }(x)+\delta _{\beta }^{\rho }\eta
^{\nu \sigma }g_{\alpha \nu }(x)]  \notag
\end{align}%
\begin{align}
& -(\eta ^{\lambda \rho }x^{\sigma }-\eta ^{\lambda \sigma }x^{\rho
})\partial _{\lambda }g_{\alpha \beta }(x)\}+\frac{1}{16\pi G}\{R\delta
_{\lambda }^{\kappa }(\eta ^{\lambda \rho }x^{\sigma }-\eta ^{\lambda \sigma
}x^{\rho })  \notag \\
& +(\frac{\partial R}{\partial \partial _{\kappa }g_{\alpha \beta }(x)}%
-\partial _{\mu }\frac{\partial R}{\partial \partial _{\kappa }\partial
_{\mu }g_{\alpha \beta }(x)}-\Gamma _{\nu \mu }^{\nu }(x)\frac{\partial R}{%
\partial \partial _{\kappa }\partial _{\mu }g_{\alpha \beta }(x)})  \notag \\
& ([-\delta _{\alpha }^{\sigma }\eta ^{\mu \rho }g_{\mu \beta }(x)+\delta
_{\alpha }^{\rho }\eta ^{\mu \sigma }g_{\mu \beta }(x)-\delta _{\beta
}^{\sigma }\eta ^{\nu \rho }g_{\alpha \nu }(x)+\delta _{\beta }^{\rho }\eta
^{\nu \sigma }g_{\alpha \nu }(x)]  \notag \\
& -(\eta ^{\lambda \rho }x^{\sigma }-\eta ^{\lambda \sigma }x^{\rho
})\partial _{\lambda }g_{\alpha \beta }(x))+\frac{\partial R}{\partial
\partial _{\kappa }\partial _{\mu }g_{\alpha \beta }(x)}\partial _{\mu
}([-\delta _{\alpha }^{\sigma }\eta ^{\mu \rho }g_{\mu \beta }(x)  \notag \\
& +\delta _{\alpha }^{\rho }\eta ^{\mu \sigma }g_{\mu \beta }(x)-\delta
_{\beta }^{\sigma }\eta ^{\nu \rho }g_{\alpha \nu }(x)+\delta _{\beta
}^{\rho }\eta ^{\nu \sigma }g_{\alpha \nu }(x)]  \notag \\
& -(\eta ^{\lambda \rho }x^{\sigma }-\eta ^{\lambda \sigma }x^{\rho
})\partial _{\lambda }g_{\alpha \beta }(x))\}  \tag{31}
\end{align}

\subsubsection{Conservation law due to \textquotedblleft4-dimensional pure
deformation\textquotedblright\ invariance}

The action (18) remains unchanged under infinitesimal \textquotedblleft
4-dimensional pure deformations\textquotedblright , which form a 6-parameter
family of infinitesimal symmetry transformations%
\begin{align}
x^{\mu }& \longmapsto \widetilde{x}^{\mu }=L_{\nu }^{\mu }x^{\nu },\text{ }%
L_{\nu }^{\mu }=\delta _{\nu }^{\mu }+\Lambda _{\nu }^{\mu },\text{ }%
|\Lambda _{\nu }^{\mu }|\ll 1,  \notag \\
\text{ }\eta _{\mu \lambda }\Lambda _{\nu }^{\lambda }& =:\Lambda _{\mu \nu
},\text{ }\Lambda _{\mu \nu }=\Lambda _{\nu \mu },  \tag{32}
\end{align}%
\begin{equation}
\delta x^{\lambda }=\Lambda _{\mu }^{\lambda }x^{\mu }=\frac{1}{2}(\eta
^{\lambda \rho }x^{\sigma }+\eta ^{\lambda \sigma }x^{\rho })\Lambda _{\rho
\sigma },  \tag{33}
\end{equation}%
\begin{equation*}
\delta u_{\xi }^{\theta }(x)=\Lambda _{\varphi }^{\theta }u_{\xi }^{\varphi
}(x)-\Lambda _{\xi }^{\eta }u_{\eta }^{\theta }(x)
\end{equation*}%
\begin{equation}
=\frac{1}{2}[\eta ^{\theta \rho }u_{\xi }^{\sigma }(x)+\eta ^{\theta \sigma
}u_{\xi }^{\rho }(x)-\delta _{\xi }^{\sigma }\eta ^{\eta \rho }u_{\eta
}^{\theta }(x)-\delta _{\xi }^{\rho }\eta ^{\eta \sigma }u_{\eta }^{\theta
}(x)]\Lambda _{\rho \sigma },  \tag{34}
\end{equation}%
\begin{equation*}
\delta g_{\alpha \beta }(x)=-\frac{1}{2}[\delta _{\alpha }^{\sigma }\eta
^{\mu \rho }g_{\mu \beta }(x)
\end{equation*}%
\begin{equation}
+\delta _{\alpha }^{\rho }\eta ^{\mu \sigma }g_{\mu \beta }(x)+\delta
_{\beta }^{\sigma }\eta ^{\nu \rho }g_{\alpha \nu }(x)+\delta _{\beta
}^{\rho }\eta ^{\nu \sigma }g_{\alpha \nu }(x)]\Lambda _{\rho \sigma } 
\tag{35}
\end{equation}%
In this case eqn.(16) reads%
\begin{equation}
\frac{\partial }{\partial x^{\kappa }}\{\sqrt{-|g(x)|}N^{\kappa \rho \sigma
}\}=0  \tag{36}
\end{equation}%
where%
\begin{align*}
& 2N^{\kappa \rho \sigma }(u(x),\partial u(x),g(x),\partial g(x),\partial
^{2}g(x))=:2N^{\kappa \rho \sigma }(x) \\
& =\mathcal{L}\delta _{\lambda }^{\kappa }(\eta ^{\lambda \rho }x^{\sigma
}+\eta ^{\lambda \sigma }x^{\rho })+\frac{\partial \mathcal{L}}{\partial
\nabla _{\kappa }u_{\xi }^{\theta }(x)}\{[\eta ^{\theta \rho }u_{\xi
}^{\sigma }(x)+\eta ^{\theta \sigma }u_{\xi }^{\rho }(x) \\
& -\delta _{\xi }^{\sigma }\eta ^{\eta \rho }u_{\eta }^{\theta }(x)-\delta
_{\xi }^{\rho }\eta ^{\eta \sigma }u_{\eta }^{\theta }(x)]-(\eta ^{\lambda
\rho }x^{\sigma }+\eta ^{\lambda \sigma }x^{\rho })\partial _{\lambda
}u_{\xi }^{\theta }(x)\}
\end{align*}%
\begin{align*}
& +\frac{1}{2}[\frac{\partial \mathcal{L}}{\partial \nabla _{\kappa }u_{\xi
}^{\theta }(x)}g^{\theta \alpha }(x)u_{\xi }^{\beta }(x)+\frac{\partial 
\mathcal{L}}{\partial \nabla _{\beta }u_{\xi }^{\theta }(x)}g^{\theta \alpha
}(x)u_{\xi }^{\kappa }(x) \\
& -\frac{\partial \mathcal{L}}{\partial \nabla _{\alpha }u_{\xi }^{\theta
}(x)}g^{\theta \kappa }(x)u_{\xi }^{\beta }(x)-\frac{\partial \mathcal{L}}{%
\partial \nabla _{\kappa }u_{\beta }^{\theta }(x)}g^{\xi \alpha }(x)u_{\xi
}^{\theta }(x) \\
& -\frac{\partial \mathcal{L}}{\partial \nabla _{\beta }u_{\kappa }^{\theta
}(x)}g^{\xi \alpha }(x)u_{\xi }^{\theta }(x)-\frac{\partial \mathcal{L}}{%
\partial \nabla _{\alpha }u_{\beta }^{\theta }(x)}g^{\xi \kappa }(x)u_{\xi
}^{\theta }(x)]\times \\
& \{[\delta _{\alpha }^{\sigma }\eta ^{\mu \rho }g_{\mu \beta }(x)+\delta
_{\alpha }^{\rho }\eta ^{\mu \sigma }g_{\mu \beta }(x)+\delta _{\beta
}^{\sigma }\eta ^{\nu \rho }g_{\alpha \nu }(x)+\delta _{\beta }^{\rho }\eta
^{\nu \sigma }g_{\alpha \nu }(x)]
\end{align*}%
\begin{align}
& +(\eta ^{\lambda \rho }x^{\sigma }+\eta ^{\lambda \sigma }x^{\rho
})\partial _{\lambda }g_{\alpha \beta }(x)\}+\frac{1}{16\pi G}\{R\delta
_{\lambda }^{\kappa }(\eta ^{\lambda \rho }x^{\sigma }-\eta ^{\lambda \sigma
}x^{\rho })  \notag \\
& -(\frac{\partial R}{\partial \partial _{\kappa }g_{\alpha \beta }(x)}%
-\partial _{\mu }\frac{\partial R}{\partial \partial _{\kappa }\partial
_{\mu }g_{\alpha \beta }(x)}-\Gamma _{\nu \mu }^{\nu }(x)\frac{\partial R}{%
\partial \partial _{\kappa }\partial _{\mu }g_{\alpha \beta }(x)})  \notag \\
& [\delta _{\alpha }^{\sigma }\eta ^{\mu \rho }g_{\mu \beta }(x)+\delta
_{\alpha }^{\rho }\eta ^{\mu \sigma }g_{\mu \beta }(x)+\delta _{\beta
}^{\sigma }\eta ^{\nu \rho }g_{\alpha \nu }(x)  \notag \\
& +\delta _{\beta }^{\rho }\eta ^{\nu \sigma }g_{\alpha \nu }(x)+(\eta
^{\lambda \rho }x^{\sigma }+\eta ^{\lambda \sigma }x^{\rho })\partial
_{\lambda }g_{\alpha \beta }(x)]  \notag \\
& -\frac{\partial R}{\partial \partial _{\kappa }\partial _{\mu }g_{\alpha
\beta }(x)}\partial _{\mu }[\delta _{\alpha }^{\sigma }\eta ^{\mu \rho
}g_{\mu \beta }(x)+\delta _{\alpha }^{\rho }\eta ^{\mu \sigma }g_{\mu \beta
}(x)+\delta _{\beta }^{\sigma }\eta ^{\nu \rho }g_{\alpha \nu }(x)  \notag \\
& +\delta _{\beta }^{\rho }\eta ^{\nu \sigma }g_{\alpha \nu }(x)+(\eta
^{\lambda \rho }x^{\sigma }+\eta ^{\lambda \sigma }x^{\rho })\partial
_{\lambda }g_{\alpha \beta }(x)]\}  \tag{37}
\end{align}

\subsubsection{Conservation law due to \textquotedblleft
scaling\textquotedblright\ invariance}

The action (18) remains unchanged under infinitesimal scaling
transformations, which form a 1-parameter family of infinitesimal symmetry
transformations

\begin{equation}
x^{\lambda}\longmapsto\widetilde{x}^{\lambda}=e^{\epsilon}x^{\lambda},\text{
\ }|\epsilon|\ll1,\text{ \ }\delta x^{\lambda}=\epsilon x^{\lambda},\text{ \
\ \ }\forall\lambda=0,1,2,3  \tag{38}
\end{equation}%
\begin{equation}
\delta u_{\xi}^{\theta}(x)=0,\text{ \ }\delta g_{\alpha\beta}(x)=-2\epsilon
g_{\alpha\beta}(x)  \tag{39}
\end{equation}
In this case eqn.(16) reads%
\begin{equation}
\frac{\partial}{\partial x^{\kappa}}\{\sqrt{-|g(x)|}S^{\kappa}\}=0  \tag{40}
\end{equation}%
\begin{align}
& S^{\kappa}(u(x),\partial u(x),g(x),\partial
g(x),\partial^{2}g(x))=:S^{\kappa}(x)  \notag \\
& =\mathcal{L}\delta_{\rho}^{\kappa}x^{\rho}-\frac{\partial\mathcal{L}}{%
\partial\nabla_{\kappa}u_{\xi}^{\theta}(x)}x^{\rho}\frac{\partial}{\partial
x^{\rho}}u_{\xi}^{\theta}(x)-\frac{1}{2}[\frac{\partial\mathcal{L}}{%
\partial\nabla_{\kappa}u_{\xi}^{\theta}(x)}g^{\theta\alpha}(x)u_{\xi}^{\beta
}(x)  \notag \\
& +\frac{\partial\mathcal{L}}{\partial\nabla_{\beta}u_{\xi}^{\theta}(x)}%
g^{\theta\alpha}(x)u_{\xi}^{\kappa}(x)-\frac{\partial\mathcal{L}}{%
\partial\nabla_{\alpha}u_{\xi}^{\theta}(x)}g^{\theta\kappa}(x)u_{\xi}^{\beta
}(x)-\frac{\partial\mathcal{L}}{\partial\nabla_{\kappa}u_{\beta}^{\theta}(x)}%
g^{\xi\alpha}(x)u_{\xi}^{\theta}(x)  \notag \\
& -\frac{\partial\mathcal{L}}{\partial\nabla_{\beta}u_{\kappa}^{\theta}(x)}%
g^{\xi\alpha}(x)u_{\xi}^{\theta}(x)+\frac{\partial\mathcal{L}}{%
\partial\nabla_{\alpha}u_{\beta}^{\theta}(x)}g^{\xi\kappa}(x)u_{\xi}^{\theta
}(x)][2g_{\alpha\beta}(x)+x^{\rho}\partial_{\rho}g_{\alpha\beta}(x)]  \notag
\\
& +\frac{1}{16\pi G}\{R\delta_{\rho}^{\kappa}x^{\rho}-[\frac{\partial R}{%
\partial\partial_{\kappa}g_{\alpha\beta}(x)}-\partial_{\mu}\frac{\partial R}{%
\partial\partial_{\kappa}\partial_{\mu}g_{\alpha\beta}(x)}-\Gamma_{\nu\mu
}^{\nu}(x)\frac{\partial R}{\partial\partial_{\kappa}\partial_{\mu}g_{\alpha%
\beta}(x)}]  \notag \\
& [2g_{\alpha\beta}(x)+x^{\rho}\partial_{\rho}g_{\alpha\beta}(x)]-\frac {%
\partial R}{\partial\partial_{\kappa}\partial_{\mu}g_{\alpha\beta}(x)}%
\partial_{\mu}[2g_{\alpha\beta}(x)+x^{\rho}\partial_{\rho}g_{\alpha\beta
}(x)]\}  \tag{41}
\end{align}

\subsubsection{Conservation law due to \textquotedblleft
skew-scaling\textquotedblright\ invariance}

The action (18) remains unchanged under infinitesimal \textquotedblleft
skew-scaling\textquotedblright\ transformations, which form a 1-parameter
family of infinitesimal symmetry transformations%
\begin{equation}
\delta x^{0}=-\epsilon ^{1}x^{0},\text{ \ }\delta x^{1}=\epsilon ^{1}x^{1},%
\text{ \ }\delta x^{2}=0,\text{ }\delta x^{3}=0\text{\ }  \tag{42}
\end{equation}%
\begin{equation}
\left[ \frac{\partial \widetilde{x}}{\partial x}\right] =\left[ 
\begin{array}{cccc}
1-\epsilon ^{1} & 0 & 0 & 0 \\ 
0 & 1+\epsilon ^{1} & 0 & 0 \\ 
0 & 0 & 1 & 0 \\ 
0 & 0 & 0 & 1%
\end{array}%
\right]  \tag{43}
\end{equation}%
\begin{equation*}
\delta u_{0}^{0}(x)=0,\text{\ }\delta u_{1}^{0}(x)=-2\epsilon
^{1}u_{1}^{0}(x),\text{ }\delta u_{2}^{0}(x)=-\epsilon ^{1}u_{2}^{0}(x),%
\text{ }\delta u_{3}^{0}(x)=-\epsilon ^{1}u_{3}^{0}(x)
\end{equation*}%
\begin{equation*}
\delta u_{0}^{1}(x)=2\epsilon ^{1}u_{0}^{1}(x),\text{\ }\delta
u_{1}^{1}(x)=0,\text{ }\delta u_{2}^{1}(x)=\epsilon ^{1}u_{2}^{1}(x),\text{ }%
\delta u_{3}^{1}(x)=\epsilon ^{1}u_{3}^{1}(x)
\end{equation*}%
\begin{equation*}
\delta u_{0}^{j}(x)=\epsilon ^{1}u_{0}^{j}(x),\text{\ }\delta
u_{1}^{j}(x)=-\epsilon ^{1}u_{1}^{j}(x),\text{ }\delta u_{k}^{j}(x)=0,\text{ 
}\forall j,k=2,3
\end{equation*}%
\begin{equation*}
\delta g_{00}(x)=2\epsilon ^{1}g_{00}(x),\text{ }\delta g_{01}(x)=0,\text{ }%
\delta g_{02}(x)=\epsilon ^{1}g_{02}(x),\text{ }\delta g_{03}(x)=\epsilon
^{1}g_{03}(x)
\end{equation*}%
\begin{equation*}
\delta g_{10}(x)=0,\text{ }\delta g_{11}(x)=-2\epsilon ^{1}g_{11}(x),\text{ }%
\delta g_{12}(x)=-\epsilon ^{1}g_{12}(x),\text{ }\delta g_{13}(x)=-\epsilon
^{1}g_{13}(x)
\end{equation*}%
\begin{equation}
\delta g_{j0}(x)=\epsilon ^{1}g_{j0}(x),\text{ }\delta g_{j1}(x)=-\epsilon
^{1}g_{j1}(x),\text{ }\delta g_{jk}(x)=0,\text{ }\forall j,k=2,3  \tag{44}
\end{equation}%
Substitute eqns.(42) and (44) into eqn.(16), \ we can get a conserved
current 
\begin{equation}
\frac{\partial }{\partial x^{\kappa }}\{\sqrt{-|g(x)|}J_{1}^{\kappa }\}=0 
\tag{45}
\end{equation}%
Here we skip the expression of $J_{1}^{\kappa }(u(x),\partial
u(x),g(x),\partial g(x),\partial ^{2}g(x))$. Similarily we can get conserved
currents $J_{2}^{\kappa }$ and $J_{3}^{\kappa }$.

Because the symmetry group of classical field in GR is an infinite
dimensional Lie group, we can get infinitely many conservation laws by using
Noether's theorem. It is interesting to note that $\tau _{\rho }^{\kappa }$
is not a (1,1) tensor, $M^{\kappa \rho \sigma }$is not a (3,0) tensor, $%
N^{\kappa \rho \sigma }$ is not a (3,0) tensor, and $S^{\kappa }$ is not a
(1,0) tensor, etc. in the following context. 
\begin{equation}
\tau _{\rho }^{\kappa }(x)\neq \frac{\partial x^{\kappa }}{\partial
y^{\lambda }}\frac{\partial y^{\sigma }}{\partial x^{\rho }}\tau _{\sigma
}^{\lambda }(y),M^{\kappa \rho \sigma }(x)\neq \frac{\partial x^{\kappa }}{%
\partial y^{\alpha }}\frac{\partial x^{\rho }}{\partial y^{\beta }}\frac{%
\partial x^{\sigma }}{\partial y^{\gamma }}M^{\alpha \beta \gamma }(y),\text{
etc.}  \tag{46}
\end{equation}%
Therefore the pseudotensor problems are not confined to gravitational
energy-momentum, they are very common in GR. The main tasks of the present
work is\ to show that these pseudotensor problems are just results from
misreading Noether's theorem.

\section{Noether's theorem revisited}

Let us start with the simplest example. Consider the following two
1-parameter local groups of coordinate transformations 
\begin{equation}
\delta x^{\rho }=\epsilon \delta _{\sigma }^{\rho },\text{ }\delta u_{\xi
}^{\theta }(x)=0,\text{ }\delta g_{\alpha \beta }(x)=0  \tag{47}
\end{equation}%
and%
\begin{equation}
\delta y^{\rho }=\epsilon \delta _{\sigma }^{\rho }\ ,\text{ }\delta u_{\xi
}^{\theta }(y)=0,\text{ }\delta g_{\alpha \beta }(y)=0.  \tag{47'}
\end{equation}%
where $\sigma $ is a fixed index. According to eqn.(26), the conservation
currents are respectively%
\begin{equation}
\tau _{\sigma }^{\lambda }(x)=\tau _{\sigma }^{\lambda }(u(x),\partial
u(x),g(x),\partial g(x),\partial ^{2}g(x))  \tag{48}
\end{equation}%
and%
\begin{equation}
\tau _{\sigma }^{\lambda }(y)=\tau _{\sigma }^{\lambda }(u(y),\partial
u(y),g(y),\partial g(y),\partial ^{2}g(y))  \tag{48'}
\end{equation}%
Since they are functions of coordinates, field quantitities and their
derivatives with the same form, they are taken as components of the same
geometrical, physical object, the canonical energy-momentum, which is
considered a pseudotensor, because of inequality (46).

The 1-parameter local groups of coordinate transformations (47) and (47') do
not correspond to the same 1-parameter group of diffeomorphisms of spacetime 
$M$ onto itself. Conversely, if a 1-parameter local group of diffeomorphisms
of spacetime $M$ onto itself is described in coordinate system $%
(x^{0},x^{1},x^{2},x^{3})$ by eqn.(47), then it would be described in
coordinate system $(y^{0},y^{1},y^{2},y^{3})$ by%
\begin{align}
\delta y^{\rho }& =\frac{\partial y^{\rho }}{\partial x^{\sigma }}\epsilon ,%
\text{ }\delta u_{\xi }^{\theta }(y)=\epsilon \lbrack \frac{\partial }{%
\partial y^{\varphi }}(\frac{\partial y^{\theta }}{\partial x^{\sigma }}%
)u_{\xi }^{\varphi }(x)-\frac{\partial }{\partial y^{\xi }}(\frac{\partial
y^{\eta }}{\partial x^{\sigma }})u_{\eta }^{\theta }(x)]\text{ }  \notag \\
\delta g_{\alpha \beta }(y)& =-\epsilon \lbrack \frac{\partial }{\partial
y^{\beta }}(\frac{\partial y^{\sigma }}{\partial x^{\sigma }})g_{\alpha
\sigma }(x)+\frac{\partial }{\partial y^{\alpha }}(\frac{\partial y^{\rho }}{%
\partial x^{\sigma }})g_{\rho \beta }(x)]\text{ } 
\tag{47\textquotedblright
}
\end{align}%
which generally is no longer \textquotedblleft coordinate
shift\textquotedblright\ (47'). Hence the conservation current in coordinate
system $(y^{0},y^{1},y^{2},y^{3})$ is no longer (47').

In general, corresponding to the 1-parameter local groups of coordinate
transformations with the same form, $\widetilde{x}^{\alpha }=f^{\alpha
}(x^{0},x^{1},x^{2},x^{3},\epsilon ),\forall \alpha =0,1,2,3$, and $%
\widetilde{y}^{\alpha }=f^{\alpha }(y^{0},y^{1},y^{2},y^{3},\epsilon
),\forall \alpha =0,1,2,3$, the expressions of $\delta x$, $\overline{\delta 
}u(x)$, $\overline{\delta }g(x)$, $\partial \overline{\delta }g(x)$ and $%
\delta y$, $\overline{\delta }u(y)$, $\overline{\delta }g(y)$, $\partial 
\overline{\delta }g(y)$ are the same, hence Noether's conservation currents
(23) $J_{x}^{\kappa }$ and its counterpart in coordinate system $%
(y^{0},y^{1},y^{2},y^{3})$, $J^{\kappa }\left[ u(y)\text{, }\partial u(y)%
\text{, }g(y)\text{, }\partial g(y)\text{, }\partial ^{2}g(y)\text{; }\delta
y\text{, }\overline{\delta }u(y)\text{, }\overline{\delta }g(y)\text{, }%
\partial \overline{\delta }g(y)\right] =:J_{y}^{\kappa }$ are functions of $%
g $, $u$, their derivatives and cooedinates, with the same form. However $%
\widetilde{x}^{\alpha }=f^{\alpha }(x^{0},x^{1},x^{2},x^{3},\epsilon )$, and 
$\widetilde{y}^{\alpha }=f^{\alpha }(y^{0},y^{1},y^{2},y^{3},\epsilon )$
don't correspond to the same 1-parameter local group of diffeomorphisms of
spacetime $M$ onto itself. While corresponding to the same 1-parameter local
group of diffeomorphisms of spacetime $M$ onto itself, the expressions of $%
\delta x,\overline{\delta }u(x),\overline{\delta }g(x),\partial \overline{%
\delta }g(x)$ and $\delta y,\overline{\delta }u(y),\overline{\delta }%
g(y),\partial \overline{\delta }g(y)$ are different, hence $J_{x}^{\kappa }$
and $J_{y}^{\kappa }$ are not functions of $g$, $u$, their derivatives and
cooedinates with the same form. But even so, I am going to show that in the
latter case $J_{x}^{\kappa }$ and $J_{y}^{\kappa }$ are components of the
same geometrical, physical object, while in the former case they are not.

\begin{theorem}
If $J_{x}^{\kappa }$ and $J_{y}^{\kappa }$ correspond to the same
1-parameter local group of diffeomorphisms of spacetime $M$ onto itself, $%
\{\phi _{\epsilon }:M\longrightarrow M|\left\vert \epsilon \right\vert \ll
1\}$, then%
\begin{equation}
\frac{\partial }{\partial x^{\kappa }}[\sqrt{-|g(x)|}J_{x}^{\kappa
}]=0\Longleftrightarrow \frac{\partial }{\partial y^{\kappa }}[\sqrt{-|g(y)|}%
J_{y}^{\kappa }]=0  \tag{49}
\end{equation}%
that is, the two conservation laws are equivalent to each other.
\end{theorem}

\begin{proof}
Using an identity on the Jacobian $\left\vert \frac{\partial y}{\partial x}%
\right\vert $ (For its proof see appendix C)%
\begin{equation}
\frac{\partial }{\partial x^{\kappa }}\left\vert \frac{\partial y}{\partial x%
}\right\vert =\left\vert \frac{\partial y}{\partial x}\right\vert \frac{%
\partial }{\partial y^{\lambda }}\left( \frac{\partial y^{\lambda }}{%
\partial x^{\kappa }}\right)  \tag{50}
\end{equation}%
we get%
\begin{align}
& \frac{\partial }{\partial x^{\kappa }}[\sqrt{-|g(x)|}J_{x}^{\kappa }]=%
\frac{\partial y^{\lambda }}{\partial x^{\kappa }}\frac{\partial }{\partial
y^{\lambda }}\left[ \sqrt{-|g(y)|}\left\vert \left\vert \frac{\partial y}{%
\partial x}\right\vert \right\vert J_{x}^{\kappa }\right]  \notag \\
& =\frac{\partial y^{\lambda }}{\partial x^{\kappa }}\frac{\partial }{%
\partial y^{\lambda }}\left[ \sqrt{-|g(y)|}J_{x}^{\kappa }\right] \left\vert
\left\vert \frac{\partial y}{\partial x}\right\vert \right\vert +\frac{%
\partial }{\partial x^{\kappa }}\left\vert \left\vert \frac{\partial y}{%
\partial x}\right\vert \right\vert \left[ \sqrt{-|g(y)|}J_{x}^{\kappa }%
\right]  \notag \\
& =\frac{\partial }{\partial y^{\lambda }}\left[ \sqrt{-|g(y)|}\frac{%
\partial y^{\lambda }}{\partial x^{\kappa }}J_{x}^{\kappa }\right]
\left\vert \left\vert \frac{\partial y}{\partial x}\right\vert \right\vert
-\left( \frac{\partial }{\partial y^{\lambda }}\frac{\partial y^{\lambda }}{%
\partial x^{\kappa }}\right) \left[ \sqrt{-|g(y)|}J_{x}^{\kappa }\right]
\left\vert \left\vert \frac{\partial y}{\partial x}\right\vert \right\vert 
\notag \\
& +\frac{\partial ^{2}y^{\lambda }}{\partial x^{\kappa }\partial x^{\alpha }}%
\frac{\partial x^{\alpha }}{\partial y^{\lambda }}\left\vert \left\vert 
\frac{\partial y}{\partial x}\right\vert \right\vert \left[ \sqrt{-|g(y)|}%
J_{x}^{\kappa }\right]  \notag \\
& =\frac{\partial }{\partial y^{\lambda }}\left[ \sqrt{-|g(y)|}\frac{%
\partial y^{\lambda }}{\partial x^{\kappa }}J_{x}^{\kappa }\right]
\left\vert \left\vert \frac{\partial y}{\partial x}\right\vert \right\vert 
\tag{51}
\end{align}%
Hence the proof is reduced to proving%
\begin{equation}
\frac{\partial y^{\lambda }}{\partial x^{\kappa }}J_{x}^{\kappa
}=J_{y}^{\lambda }  \tag{52}
\end{equation}%
In fact, it's easy to show that $\delta x^{\kappa }$ is a vector, $\overline{%
\delta }u_{\xi }^{\theta }(x)=\widetilde{u}_{\xi }^{\theta }(x)-u_{\xi
}^{\theta }(x)=\delta u_{\xi }^{\theta }(x)-\delta x^{\rho }\partial _{\rho
}u_{\xi }^{\theta }(x)$ is a $(1,1)$-tensor, and $\overline{\delta }%
g_{\alpha \beta }(x)=\widetilde{g}_{\alpha \beta }(x)-g_{\alpha \beta }(x)$ $%
=\delta g_{\alpha \beta }(x)-\delta x^{\rho }\partial _{\rho }g_{\alpha
\beta }(x)$ is a (0,2)-tensor. Hence the terms in the first brace of $%
J_{x}^{\kappa }$ eqn.(23) are vectors, and the first term in the second
brace is a vector too. The rest of the terms in the second brace are not
vectors individually. However, their sum is a vector. This can be proven
straightforwardly, though tediously. (See Appendix D)
\end{proof}

\begin{remark}
Theorem 3 says, corresponding to the same 1-parameter local group of
diffeomorphisms of spacetime $M$ onto itself, the continuity equations
written in different coordinate systems are equivalent to one another. Our
proof result (eqn.(52)) tells us more than Eqn.(49). It says,

\begin{remark}
The Noether conservation current is\ a vector field over spacetime
independent of coordinates. It should be the density, and current density of
some scalar. Therefore, the two continuity equations in eqn.(49) are the
concervation law of the same scalar.
\end{remark}
\end{remark}

\begin{remark}
All the Noether's conserved quantities are scalars.\ 

\begin{remark}
The pseudotensor problem, hence the non-localizability problem in GR are
really a result of misreading Noether's theorem and mistaking different
geometrical physical objects as one and the same.
\end{remark}
\end{remark}

Einstein, Landau and Lifshitz, et al. did not use Noether's theorem to get
their continuity equation for the total energy-momentum%
\begin{equation*}
\frac{\partial }{\partial x^{\alpha }}[\sqrt{-|g(x)|}(T^{\alpha \beta
}(x)+t^{\alpha \beta }(x))]=0,
\end{equation*}%
\begin{equation}
\forall \text{ coordinate systems }(x^{0},x^{1},x^{2},x^{3})\text{ of }M%
\text{, and indices }\beta .  \tag{53 }
\end{equation}%
where $t^{\alpha \beta }(x)$ is the gravitational energy-momentum
pseudotensor. By using the following theorem, I will show that for each
specified $(x^{0},x^{1},x^{2},x^{3})$ and $\beta $, eqn.(53) determines a
conservation law for some scalar depending on $(x^{0},x^{1},x^{2},x^{3})$
and $\beta $.

\begin{theorem}
Suppose $t$ is an $(r+1,s)$-tensor field on spacetime $M$, $(\xi ^{0}$, $\xi
^{1}$, $\xi ^{2}$, $\xi ^{3})$ is a coordinate system of $M$. If for some
indices $\alpha _{1},\ldots ,\alpha _{r},\beta _{1},\ldots ,\beta _{s}$,%
\begin{equation}
\frac{\partial }{\partial \xi ^{\lambda }}\ [\sqrt{-|g(\xi )|}t_{\ \ \ \ \ \
\ \beta _{1}\cdots \beta _{s}}^{\lambda \alpha _{1}\cdots \alpha _{r}}(\xi
)]=0  \tag{54}
\end{equation}%
then for every coordinate system $(x^{0},x^{1},x^{2},x^{3})$ of $M$,%
\begin{equation}
\frac{\partial }{\partial x^{\lambda }}\ [\sqrt{-|g(x)|}j^{\lambda }(x)]=0 
\tag{55}
\end{equation}%
where $j$ is a tangent field on $M$ defined as follows%
\begin{equation}
j(dx^{\lambda })=t(dx^{\lambda },d\xi ^{\alpha _{1}},\ldots ,d\xi ^{\alpha
_{r}},\frac{\partial }{\partial \xi ^{\beta _{1}}},\ldots ,\frac{\partial }{%
\partial \xi ^{\beta _{s}}})  \tag{56}
\end{equation}
\end{theorem}

\begin{proof}
\begin{align*}
& \frac{\partial }{\partial \xi ^{\lambda }}\ [\sqrt{-|g(\xi )|}t_{\ \ \ \ \
\ \ \beta _{1}\cdots \beta _{s}}^{\lambda \alpha _{1}\cdots \alpha _{r}}(\xi
)]=\frac{\partial }{\partial \xi ^{\lambda }}\ [\sqrt{-|g(\xi )|}j_{\
}^{\lambda }(\xi )] \\
& =\frac{\partial }{\partial \ \xi ^{\lambda }}[\left\vert \frac{\partial x}{%
\partial \xi }\right\vert \frac{\partial \xi ^{\lambda }}{\partial x^{\mu }}%
\sqrt{-|g(x)|}j^{\mu }(x)\ ]
\end{align*}%
\begin{equation*}
=\left\vert \frac{\partial x}{\partial \xi }\right\vert \frac{\partial \xi
^{\lambda }}{\partial x^{\mu }}\frac{\partial }{\partial \ \xi ^{\lambda }}[%
\sqrt{-|g(x)|}j^{\mu }(x)\ ]+\frac{\partial }{\partial \xi ^{\lambda }}%
[\left\vert \frac{\partial x}{\partial \xi }\right\vert \frac{\partial \xi
^{\lambda }}{\partial x^{\mu }}]\sqrt{-|g(x)|}j^{\mu }(x)
\end{equation*}%
\begin{equation}
=\left\vert \frac{\partial x}{\partial \xi }\right\vert \frac{\partial }{%
\partial x^{\mu }}[\sqrt{-|g(x)|}j^{\mu }(\ x\ )].\text{ }\because \frac{%
\partial }{\partial \xi ^{\lambda }}[\left\vert \frac{\partial x}{\partial
\xi }\right\vert \frac{\partial \xi ^{\lambda }}{\partial x^{\mu }}]=0 
\tag{57}
\end{equation}
\end{proof}

\begin{corollary}
When $r=s=0$, theorem 3 tells us: Suppose $j$ is a vector field on spacetime 
$M$, and $(\xi ^{0},\xi ^{1},\xi ^{2},\xi ^{3})$ is a coordinate system of $%
M $. if 
\begin{equation}
\frac{\partial }{\partial \xi ^{\lambda }}\ [\sqrt{-|g(\xi )|}j_{\
}^{\lambda }(\xi )]=0  \tag{58}
\end{equation}%
then for all coordinate systems $(x^{0},x^{1},x^{2},x^{3})$%
\begin{equation}
\frac{\partial }{\partial x^{\lambda }}\ [\sqrt{-|g(x)|}j_{\ }^{\lambda
}(x)]=0  \tag{59}
\end{equation}
\end{corollary}

\begin{example}
Let us get back to eqn.(53). For a specified pair of $%
(x^{0},x^{1},x^{2},x^{3})$ and $\beta $, define vector field $J$%
\begin{equation*}
J^{\lambda }(y)=T(dy^{\lambda },dx^{\beta })+\frac{\partial y^{\lambda }}{%
\partial x^{\alpha }}t^{\alpha \beta }(x),
\end{equation*}%
\begin{equation}
\forall \ \text{coordinate systems }(y^{0},y^{1},y^{2},y^{3})\text{ of }M 
\tag{60}
\end{equation}%
Then we have the following conservation law for some scalar depending on $%
(x^{0},x^{1},x^{2},x^{3})$ and $\beta $.%
\begin{equation}
\frac{\partial }{\partial y^{\alpha }}[\sqrt{-|g(y)|}J^{\alpha
}(y)]=0,\forall \ \text{coordinate systems }(y^{0},y^{1},y^{2},y^{3})\text{
of }M  \tag{61}
\end{equation}%
Therefore, eqn.(53) plus each pair of $(x^{0},x^{1},x^{2},x^{3})$ and $\beta 
$, determines a conservation law of a scalar. We have infinitely many such
conserved scalars. Comparing eqn.(60) and $T^{\alpha \beta }(y)+t^{\alpha
\beta }(y)$, one sees the former is addition of two vector fields, while the
latter is considered addition of a tensor and a pseudotensor field. So, the
new perspective enables us to get rid of the embarrassing situation:
accepting the addition of a tensor and a pseudotensor. It is absurd
geometrically.
\end{example}

\begin{example}
Now let us consider, say, eqn.(30). For specified $(x^{0},x^{1},x^{2},x^{3})$
and $\rho ,\sigma $, define vector field $I$%
\begin{equation}
I^{\lambda }(y)=\frac{\partial y^{\lambda }}{\partial x^{\kappa }}M^{\kappa
\rho \sigma }(x),\forall \ \text{coordinate systems }%
(y^{0},y^{1},y^{2},y^{3})\text{ of }M  \tag{62}
\end{equation}%
Then we have the following conservation law of some scalar depending on $%
(x^{0}$, $x^{1}$, $x^{2}$, $x^{3})$ and $\rho $, $\sigma $.%
\begin{equation}
\frac{\partial }{\partial y^{\alpha }}[\sqrt{-|g(y)|}I^{\alpha
}(y)]=0,\forall \ \text{coordinate systems }(y^{0},y^{1},y^{2},y^{3})\text{
of }M  \tag{63}
\end{equation}%
Therefore, eqn.(30) plus each triplet of $(x^{0},x^{1},x^{2},x^{3})$ and $%
\rho ,\sigma $, determines a conservation law of a scalar. We have
infinitely many such conserved scalars.
\end{example}

\begin{example}
In special relativity (SR), the spacetime manifold $M$ is the Minkowski
space. The action functional of a $(1,1)-$tensor field is expressed in any
inertial coordinate system $(x^{0},x^{1},x^{2},x^{3})$ as%
\begin{equation}
A[u]=\int_{x(\Omega )}d^{4}x\mathcal{L}(\eta ,u(x),\partial u(x))  \tag{64}
\end{equation}%
where $\mathcal{L}(\eta ,u(x),\partial u(x))$ is the sum of a few scalars
obtained by contracting $\eta ,u(x),$and $\partial u(x)$, and multiplying
the contractions by proper coefficients. The symmetry group of this
dynamical system is the Poincar\'{e} group $\mathcal{P}$. Let $\{\phi
_{\epsilon }:M\longrightarrow M|\epsilon \in \mathbb{R}\}$ be a
1-dimensional subgroup of $\mathcal{P}$. The corresponding infinitesimal
coordinate transformations are 
\begin{align}
x^{\lambda }(p)& \longmapsto \widetilde{x}^{\lambda }(p)=x^{\lambda }(\phi
_{\epsilon }(p))=x^{\lambda }(p)+\delta x^{\lambda }(p),  \notag \\
u_{\varphi }^{\theta }(x)& \longmapsto \widetilde{u}_{\varphi }^{\theta }(%
\widetilde{x})=(\phi _{\epsilon \ast }u)_{\varphi }^{\theta }(\widetilde{x}%
)=u_{\varphi }^{\theta }(x)+\delta u_{\varphi }^{\theta }(x),\text{ }  \notag
\\
\overline{\delta }u_{\varphi }^{\theta }(x)& =(\phi _{\epsilon \ast
}u)_{\varphi }^{\theta }(x)-u_{\varphi }^{\theta }(x)=\delta u_{\varphi
}^{\theta }(x)-\delta x^{\lambda }\partial _{\lambda }u_{\varphi }^{\theta
}(x)  \tag{65}
\end{align}%
The continuity eqn.(14) now reads%
\begin{equation}
\partial _{\lambda }[\delta x^{\lambda }\mathcal{L}+\overline{\delta }%
u_{\varphi }^{\theta }(x)\frac{\partial \mathcal{L}}{\partial \partial
_{\lambda }u_{\varphi }^{\theta }(x)}]=0  \tag{66}
\end{equation}%
where $(x^{0},x^{1},x^{2},x^{3})$ is an inertial coordinate system. Rewrite
eqn.(66) as%
\begin{equation}
\partial _{\lambda }\{\sqrt{-|g(x)|}[\delta x^{\lambda }\mathcal{L}+%
\overline{\delta }u_{\varphi }^{\theta }(x)\frac{\partial \mathcal{L}}{%
\partial \nabla _{\lambda }u_{\varphi }^{\theta }(x)}]\}=0  \tag{67}
\end{equation}%
Note that $\delta x^{\lambda }\mathcal{L}+\overline{\delta }u_{\varphi
}^{\theta }(x)\frac{\partial \mathcal{L}}{\partial \nabla _{\lambda
}u_{\varphi }^{\theta }(x)}$ is a vector field. Using theorem 8, we see the
continuity eqn.(67) holds in any coordinate system. It is a conservation law
of some scalar. So, all the conserved quantities in SR are scalars under
general coordinate transformations too.
\end{example}

There have been lots of elegant presentations of Noether's theorem in the
literature since 1918. But it has never been used to disprove the long
existing pseudotensor and non-localizability problem, which is one of the
fundamental issues in GR. What has prevented people to do so? Their reasons
are:

(i) According to the principle of general covariance, the same geometrical
physical object should be expressed in all coordinate systems the same way;
hence inequality (46) is not a wrong comparison.

(ii) The non-localizability of gravitational energy is the consequence of
the physical principle of equivalence, and it is also the consequence of the
following mathematical fact: For any geodesic $G$ in spacetime, one can
always choose coordinates such that $g|_{p}=\eta |_{p}$, and all the
Chritoffel symbols $\Gamma _{\beta \gamma }^{\alpha }|_{p}=0$, $\forall p\in
G$. Hence\ it is inherent in the theory of general relativity.

Let us examine these reasons in the following.

\section{Principle of general covariance, equivalence principle and
pseudotensor, non-localizability}

\subsection{Principle of general covariance}

According to Einstein, "What we call physics comprises that group of natural
sciences which base their concepts on measurements; and whose concepts and
propositions lend themselves to mathematical formulation. Its realm is
accordingly defined as that part of the sum total of our knowledge which is
capable of being expressed in mathematical terms." Therefore, to study
physical processes, one has to choose some reference coordinate systems
first. The physical laws are objective. If their expressions depend on the
reference coordinate systems chosen by individuals, they are certainly not
being formulated properly. Therefore the principle of general covariance
requires all the physical laws be expressed in different reference
coordinate systems the same way. It is important, however, to distinguish
general physical laws and concrete physical processes (or concrete physical
quantities). The principle of general covariance also requires any concrete
physical process be observed (or any concrete physical quantity be measured)
from different reference coordinate systems the same way (All the reference
coordinate systems are the same good for observing and measuring). However,
this does not mean that a concrete physical process (or a concrete physical
quantity) should have the same relation to different reference coordinate
systems.

To illustrate the above idea, let us consider the following examples.

Einstein's field equation (20) is a general law of physics. It has the same
form in all reference coordinate systems. The energy-momentun tensor of
matter $T^{\alpha \beta }$ (21) is part of Einstein's field equation, hence
it has the same form in all reference coordinate systems. For a given
dynamical system in GR, there is only one energy-momentun tensor of matter,
which is a symmetrical (2,0)-tensor field, independant of coordinates, but
not "conserved" (See Appendix A).

While for a given dynamical system in GR, any vector field on spacetime
generates a 1-parameter local group of diffeomorphisms of spacetime $M$ onto
itself, and determines a conserved scalar independently of the coordinates.
This conserved scalar is a concrete physical quantity. Its expressions in
terms of coordinates, $g,u$ and derivatives of $g,u$, in different
coordinate systems are different (A concrete physical quantity has different
relations to different coordinate systems); but it is measured from all
coordinate systems the same way (by using eqn.(23)). All the conserved
scalars (including infinite canonical energy-momentums of the dynamical
system) are concrete physical quantities. They are different from the matter
energy-momentum $T^{\alpha \beta }$, which is part of a general law of
physics. The latter's expression in terms of coordinates, $g,u$ and
derivatives of $g,u$, does not change with coordinate systems; while the
former's expressions do.

\subsection{Equivalence principle}

Let us examine the following example, Landau-Lifshitz's gravitational
energy-momentum pseudotensor[1]%
\begin{align}
t^{\mu \nu }(x)& =\frac{1}{16\pi G}\{[2\Gamma _{\lambda \alpha }^{\beta
}(x)\Gamma _{\beta \rho }^{\rho }(x)-\Gamma _{\lambda \rho }^{\beta
}(x)\Gamma _{\alpha \beta }^{\rho }(x)-\Gamma _{\beta \lambda }^{\beta
}(x)\Gamma _{\alpha \rho }^{\rho }(x)]  \notag \\
& [g^{\mu \lambda }(x)g^{\nu \alpha }(x)-g^{\mu \nu }(x)g^{\lambda \alpha
}(x)]+g^{\mu \lambda }(x)g^{\alpha \beta }(x)[\Gamma _{\lambda \rho }^{\nu
}(x)\Gamma _{\alpha \beta }^{\rho }(x)+  \notag \\
& \Gamma _{\alpha \beta }^{\nu }(x)\Gamma _{\lambda \rho }^{\rho }(x)-\Gamma
_{\beta \rho }^{\nu }(x)\Gamma _{\lambda \alpha }^{\rho }(x)-\Gamma
_{\lambda \alpha }^{\nu }(x)\Gamma _{\beta \rho }^{\rho }(x)]  \notag \\
& +g^{\nu \lambda }(x)g^{\alpha \beta }(x)[\Gamma _{\lambda \rho }^{\mu
}(x)\Gamma _{\alpha \beta }^{\rho }(x)+\Gamma _{\alpha \beta }^{\mu
}(x)\Gamma _{\lambda \rho }^{\rho }(x)-\Gamma _{\beta \rho }^{\mu }(x)\Gamma
_{\lambda \alpha }^{\rho }(x)  \notag \\
& -\Gamma _{\lambda \alpha }^{\mu }(x)\Gamma _{\beta \rho }^{\rho
}(x)]+g^{\lambda \alpha }(x)g^{\beta \rho }(x)[\Gamma _{\lambda \beta }^{\mu
}(x)\Gamma _{\alpha \rho }^{\nu }(x)-\Gamma _{\lambda \alpha }^{\mu
}(x)\Gamma _{\beta \rho }^{\nu }(x)]\}.  \tag{68}
\end{align}%
The Einstein field equation has the following solution: $T^{\alpha \beta
}(x)\equiv 0$, $g_{\alpha \beta }(x)\equiv \eta _{\alpha \beta }$. The
spacetime is the Minkowski space, and we can choose a coordinate system of
inertia $(x^{0},x^{1},x^{2},x^{3})$. In this coordinate system, 
\begin{equation}
t^{\mu \nu }(x)\equiv 0,\forall 0\leqslant \mu ,\nu \leqslant 3.  \tag{69}
\end{equation}%
Let us switch to coordinate system $(y^{0},y^{1},y^{2},y^{3})=:(t,r,\theta
,\varphi )$, such that%
\begin{align}
t& =x^{0},r=\sqrt{(x^{1})^{2}+(x^{2})^{2}+(x^{3})^{2}},  \notag \\
\theta & =\cos ^{-1}\frac{x^{3}}{\sqrt{(x^{1})^{2}+(x^{2})^{2}+(x^{3})^{2}}},
\notag \\
\varphi & =\tan ^{-1}\frac{x^{2}}{x^{1}}.  \tag{70}
\end{align}%
In this spherical polar coordinates, 
\begin{equation}
t^{00}(y)=\frac{-(3+\cot ^{2}\theta )}{8\pi Gr^{2}}<0.  \tag{71}
\end{equation}%
Note that the coordinate system $(x^{0},x^{1},x^{2},x^{3})$ and the
coordinate system $(y^{0},y^{1},y^{2},y^{3})=:(t,r,\theta ,\varphi )$, are
not in relative motion. The coordinate trnsformation between them is only a
purely spacial one. Yet for any spacetime point $p\in M,$we have $t^{\mu \nu
}(x)=0;$and $t^{00}(y)<0$. This counter-example shows that the
non-localizability of gravitational energy-momentum in GR, can not be
physically attributed to the local indistinguishability of inertial force
and gravity.

One might argue that for any time-like geodesic $G(\tau )$ in spacetme, one
can always switch to the freely falling nonspinning observor's proper
coordinate system, so that the metric $g|_{p}=\eta $, and all the
Christoffel simbols $\Gamma |_{p}=0$, for all $p\in G(\tau )$. For this
observor there is no gravitational field around him. And this explains
non-localizability of gravitational energy-momentum.

It is worth noting, however, that the term "gravitational field" means
different things in pre-GR physics and in GR. In pre-GR physics, it means
gravitational force field described by field strength (gravity
acceleration). While in GR, it means spacetime bending described by the
metric field. These two descriptions are not equivalent to each other. The
former is only an approximation, effective in some aspect in the limiting
case of weak field and low velocity. The notion of gravitational force field
has a fatal difficulty: its field energy density is negative without lower
bound. While spacetime metric field explains the equivalence principle
perfectly. Therefore, there is no room for gravitational force field in GR.
Metric field description contains all physics from the equivalence
principle, so, as Synge has suggested, the midwife of GR be buried with
appropriate honours[4].

According GR, any freely falling body's world-line is a geodesic in
spacetime. That is the physics. Coordinate transformatiom does not change
geodesic, it does not change physics either.

According to modern differential geometry, the connection on a vector bundle
is a coordinate free notion. For any given point $p$ on the base space,
there always exist local frame fields such that in them the connection
matrix at $p$ vanishes. Consider the 2-dimensional sphere in $\mathbb{R}^{3}$
(the surface of our planet). If we switch to the longitude-latitude local
coordinate system, the Christoffel symbols at all points on the equator
vanish. The geometry does change at all. And it has nothing to do with
gravitation! Even though connection matrix transfomes under local frame
field transformatons in a way different from the tensorial way of
transformation, connection is still a notion independent of coordinates. In
geometry, an affine connection space or a general Riemannian space is
locally flat, if and only if for any point $p$ in the space, there exist
local coordinate system $(x^{0},x^{1},x^{2},x^{3})$ and open neighbourhood $%
U $ of $p$, such that $\Gamma |_{q}=0,\forall q\in U$. It should not be read
as: for any point $p$ in the space, there exist local coordinate system $%
(x^{0},x^{1},x^{2},x^{3})$ such that $\Gamma |_{p}=0$.

Both geometry and physics pursue objective scientific truth which does not
depend on individual's subjective will. To study physical and geometrical
problems, one has to choose a coordinate system first. A good geometrical
notion should be independent of coordinates, a good physical notion should
be independent of coordinates either. That is the reason why the notion of
gravitational force field should be abandoned in GR.

Now we are in a position to review the issues of conservation of
energy-momentum and the gravitational energy-momentum in GR. These will be
done in a later paper.

\appendix

\section{Appendix}

\begin{proposition}
Suppose $T^{\mu\nu}(x)$ is a (2,0)-tensor field over spacetime. If\ in all
coordinate systems $(x^{0},x^{1},x^{2},x^{3})$ 
\begin{equation}
\frac{\partial}{\partial x^{\mu}}[\sqrt{-|g(x)|}T^{\mu\nu}(x)]=0  \tag{A1}
\end{equation}
then $T^{\mu\nu}(x)=-T^{\nu\mu}(x)$.
\end{proposition}

\begin{proof}
Using the following identity%
\begin{align}
\frac{\partial}{\partial x^{\mu}}[\sqrt{-|g(x)|}T^{\mu\nu}(x)] & =\sqrt{%
-|g(x)|}\frac{\partial}{\partial x^{\mu}}T^{\mu\nu}(x)+\sqrt {-|g(x)|}%
\Gamma_{\mu\lambda}^{\mu}(x)T^{\lambda\nu}(x)  \notag \\
& =\sqrt{-|g(x)|}\nabla_{\mu}T^{\mu\nu}(x)-\sqrt{-|g(x)|}\Gamma_{\mu\lambda
}^{\nu}(x)T^{\mu\lambda}(x)  \tag{A2}
\end{align}
and the condition given above Eqn.(A1), one gets, in all coordinate systems $%
(x^{0},x^{1},x^{2},x^{3})$ 
\begin{equation}
\nabla_{\mu}T^{\mu\nu}(x)=\Gamma_{\mu\lambda}^{\nu}(x)T^{\mu\lambda}(x) 
\tag{A3}
\end{equation}
For any given point $p$ of spacetime, there exists local inertial coordinate
systems of $p$. We get, in any local inertial coordinate system of $p$, say, 
$(y^{0},y^{1},y^{2},y^{3})$%
\begin{equation}
\nabla_{\mu}T^{\mu\nu}(y)|_{p}=0  \tag{A4}
\end{equation}
Due to the tenson property of $\nabla_{\mu}T^{\mu\nu}$, this is true for all
coordinate systems. Due to the arbitrariness of point $p\in M$, we get for
any coordinate system, say, $(x^{0},x^{1},x^{2},x^{3})$, and everywhere in
spacetime 
\begin{equation}
\nabla_{\mu}T^{\mu\nu}(x)=\Gamma_{\mu\lambda}^{\nu}(x)T^{\mu\lambda}(x)=0 
\tag{A5}
\end{equation}
In the local inertial coordinate system of $p$, $(y^{0},y^{1},y^{2},y^{3})$,
let%
\begin{equation*}
\lbrack T^{\mu\nu}(y)+T^{\nu\mu}(y)]_{p}=:C^{\mu\nu}=C^{\nu\mu}
\end{equation*}
Transform to a new coordinate system $(z^{0},z^{1},z^{2},z^{3})$ such that%
\begin{equation*}
y^{\lambda}\equiv(y^{\lambda})_{p}+x^{\lambda}+\sum\nolimits_{\mu\nu}\frac {1%
}{2}C^{\mu\nu}x^{\mu}x^{\nu}
\end{equation*}
In this new coordinate system%
\begin{equation*}
T^{\mu\nu}(z)|_{p}=T^{\mu\nu}(y)|_{p}
\end{equation*}%
\begin{equation*}
\Gamma_{\mu\lambda}^{\nu}(z)|_{p}=C^{\mu\nu}
\end{equation*}%
\begin{equation*}
\lbrack\Gamma_{\mu\lambda}^{\nu}(z)T^{\mu\lambda}(z)]_{p}=\frac{1}{2}%
\sum\nolimits_{\mu\nu}(C^{\mu\nu})^{2}=0
\end{equation*}
Hence%
\begin{equation*}
\lbrack T^{\mu\nu}(z)+T^{\nu\mu}(z)]_{p}=0
\end{equation*}
For any coordinate system, say, $(x^{0},x^{1},x^{2},x^{3})$, and everywhere
in spacetime%
\begin{equation}
T^{\mu\nu}(x)+T^{\nu\mu}(x)=0  \tag{A6}
\end{equation}
\end{proof}

\begin{proposition}
Suppose $T^{\mu \nu }(x)$ is a skew symmetric (2,0)-tensor field over
spacetime $M$. If\ in some coordinate system $(x^{0},x^{1},x^{2},x^{3})$%
\begin{equation}
\frac{\partial }{\partial x^{\mu }}[\sqrt{-|g(x)|}T^{\mu \nu }(x)]=0 
\tag{A7}
\end{equation}%
then it holds for all the coordinate systems.
\end{proposition}

Eqn.(A1) is considered the continuity equation for some tangent vector $%
P^{\nu }$ whose density and current density is $T^{\mu \nu }$, by Einstein,
Landau, et al. In a flat spacetime, we can talk about the sum of $(r,s)$%
-tensors distribued at different spacetime points. But in a curved
spacetime, we can't, unless $r=s=0$. In a curved coordinate system $%
(x^{0},x^{1},x^{2},x^{3})$, the expression $\int\nolimits_{x(\Sigma
)}ds_{\mu }(x)[\sqrt{-|g(x)|}T^{\mu \nu }(x)]$ ($\Sigma $ is a space-like
hypersurface) is not the $\nu $-component of the sum vector $P$ over $\Sigma 
$\ which can not be difined in a curved spacetime.

\section{Appendix}

The Lagragian density of classical fields $L$, is a function of the
coordinates, field quatities, and their derivatives of up to the $N$-th
order. Because not all the arguments are independent, such as $%
\partial_{\mu}\partial_{\nu}u_{\eta}^{\xi}(x)=\partial_{\nu}\partial_{%
\mu}u_{\eta}^{\xi}(x)$, $g_{\alpha\beta}(x)=g_{\beta\alpha}(x)$, etc., there
are infinitely many different function forms for $L$. This causes
indefiniteness of derivatives, such as $\frac{\partial}{\partial\partial_{%
\mu}\partial_{\nu}u_{\eta}^{\xi }(x)}L,\frac{\partial}{\partial
g_{\alpha\beta}(x)}L$. If we drop all the redundent variables, then the
Einstein summation convention can no longer be used, and the expressions
will become awfully complicated, especially for a large $N$. In order to
keep the formulaes neat, physists usually treat it in a different way. Here
we will illustrate their method by using the lagrangian density for vacuum
Einstein's equation, $R$ (Ricci's scalar curvature).

$R$ is a function of 16 $g_{\alpha\beta}(x)$'s, 64 $\partial_{\mu}g_{\alpha%
\beta}(x)$'s, and 256 $\partial_{\mu}\partial_{\nu}g_{\alpha\beta }(x)$'s.
Because $g_{\alpha\beta}(x)=g_{\beta\alpha}(x)$, $\partial_{\mu
}g_{\alpha\beta}(x)=\partial_{\mu}g_{\beta\alpha}(x)$ and $\partial_{\mu
}\partial_{\nu}g_{\alpha\beta}(x)=\partial_{\nu}\partial_{\mu}g_{\alpha\beta
}(x)=\partial_{\mu}\partial_{\nu}g_{\beta\alpha}(x)=\partial_{\nu}\partial_{%
\mu}g_{\beta\alpha}(x)$, there are only 150 independent variables among
them. We will choose 10 $g_{\alpha\beta}(x)$'s, 40 $\partial_{\mu
}g_{\alpha\beta}(x)$'s, and 100 $\partial_{\mu}\partial_{\nu}g_{\alpha\beta
}(x)$'s $(\alpha\leq\beta,\mu\leq\nu)$, for the independent variables. As a
function of 336 variables (As a function defined on a 336-demensional
domain), $R$ can take infinite different forms, say, $\varphi,\psi,\ldots$%
When restricted to the 150-dimensional \textquotedblleft
sub-domain\textquotedblright\ $D$, all of them are the same function of 150
variables. 
\begin{equation}
R|_{D}=\varphi|_{D}=\psi|_{D}=\ldots  \tag{B1}
\end{equation}
Substituting the 150 independent variables for all the variables in $%
\varphi,\psi,\ldots$, we get a unique function%
\begin{equation}
\underline{R}(g_{\alpha\beta}(x),\partial_{\mu}g_{\alpha\beta}(x),\partial
_{\mu}\partial_{\nu}g_{\alpha\beta}(x)),\text{ }(\alpha\leq\beta,\mu\leq\nu)
\tag{B2}
\end{equation}
Substituting $\frac{1}{2}(g_{\alpha\beta}(x)+g_{\beta\alpha}(x))$, $\frac {1%
}{2}(\partial_{\mu}g_{\alpha\beta}(x)+\partial_{\mu}g_{\beta\alpha}(x))$,
and $\frac{1}{4}(\partial_{\mu}\partial_{\nu}g_{\alpha\beta}(x)+\partial_{%
\mu
}\partial_{\nu}g_{\beta\alpha}(x)+\partial_{\nu}\partial_{\mu}g_{\alpha\beta
}(x)+\partial_{\nu}\partial_{\mu}g_{\beta\alpha}(x))$ for $%
g_{\alpha\beta}(x) $, $\partial_{\mu}g_{\alpha\beta}(x)$, and $%
\partial_{\mu}\partial_{\nu }g_{\alpha\beta}(x)$ in $\underline{R}$,
respectively, we get a unique function of all 336 variables, denoted by $%
R(g(x),\partial g(x),\partial ^{2}g(x))$. This \textquotedblleft
standard\textquotedblright\ $R(g(x),\partial g(x),\partial^{2}g(x))$ has the
following property.%
\begin{align}
\frac{\partial R}{\partial g_{\alpha\beta}(x)} & =\frac{\partial R}{\partial
g_{\beta\alpha}(x)}=\frac{1}{2}\frac{\partial\underline{R}}{\partial
g_{\alpha\beta}(x)},\alpha<\beta  \notag \\
\frac{\partial R}{\partial g_{\alpha\alpha}(x)} & =\frac{\partial \underline{%
R}}{\partial g_{\alpha\alpha}(x)}  \tag{B3}
\end{align}%
\begin{align}
\frac{\partial R}{\partial\partial_{\mu}g_{\alpha\beta}(x)} & =\frac {%
\partial R}{\partial\partial_{\mu}g_{\beta\alpha}(x)}=\frac{1}{2}\frac{%
\partial\underline{R}}{\partial\partial_{\mu}g_{\alpha\beta}(x)},\alpha<\beta
\notag \\
\frac{\partial R}{\partial\partial_{\mu}g_{\alpha\alpha}(x)} & =\frac{%
\partial\underline{R}}{\partial\partial_{\mu}g_{\alpha\alpha}(x)}  \tag{B4}
\end{align}%
\begin{align}
\frac{\partial R}{\partial\partial_{\mu}\partial_{\nu}g_{\alpha\beta}(x)} & =%
\frac{\partial R}{\partial\partial_{\mu}\partial_{\nu}g_{\beta\alpha}(x)}=%
\frac{\partial R}{\partial\partial_{\nu}\partial_{\mu}g_{\alpha\beta}(x)}= 
\notag \\
\frac{\partial R}{\partial\partial_{\nu}\partial_{\mu}g_{\beta\alpha}(x)} & =%
\frac{1}{4}\frac{\partial\underline{R}}{\partial\partial_{\mu}g_{\alpha\beta
}(x)},\alpha<\beta,\mu<\nu  \tag{B5}
\end{align}%
\begin{align}
\frac{\partial R}{\partial\partial_{\mu}\partial_{\nu}g_{\alpha\alpha}(x)} &
=\frac{\partial R}{\partial\partial_{\nu}\partial_{\mu}g_{\alpha\alpha}(x)}=%
\frac{1}{2}\frac{\partial\underline{R}}{\partial\partial_{\mu}\partial_{%
\nu}g_{\alpha\alpha}(x)},\mu<\nu  \notag \\
\frac{\partial R}{\partial\partial_{\mu}\partial_{\mu}g_{\alpha\beta}(x)} & =%
\frac{\partial R}{\partial\partial_{\mu}\partial_{\mu}g_{\alpha\beta}(x)}=%
\frac{1}{2}\frac{\partial\underline{R}}{\partial\partial_{\mu}\partial_{%
\mu}g_{\alpha\beta}(x)},\alpha<\beta  \notag \\
\frac{\partial R}{\partial\partial_{\mu}\partial_{\mu}g_{\alpha\alpha}(x)} &
=\frac{\partial\underline{R}}{\partial\partial_{\mu}\partial_{\mu}g_{\alpha%
\alpha}(x)}  \tag{B6}
\end{align}

When calculating the derivatives of $R$, we pretend that all its 336
variables are independent. Thus the indefiniteness problem no longer exists.

From (B1), we have%
\begin{equation*}
\delta\varphi|_{D}=\delta\psi|_{D}
\end{equation*}
While%
\begin{align*}
\delta\varphi|_{D} & =[\frac{\partial\varphi}{\partial g_{\alpha\beta}(x)}%
\delta g_{\alpha\beta}(x)+\frac{\partial\varphi}{\partial\partial_{\mu
}g_{\alpha\beta}(x)}\delta\partial_{\mu}g_{\alpha\beta}(x)+\frac {%
\partial\varphi}{\partial\partial_{\mu}\partial_{\nu}g_{\alpha\beta}(x)}%
\delta\partial_{\mu}\partial_{\nu}g_{\alpha\beta}(x)]|_{D} \\
& =[\sum\nolimits_{\alpha}\frac{\partial\varphi}{\partial g_{\alpha\alpha
}(x)}\delta g_{\alpha\alpha}(x)+\sum\nolimits_{\alpha<\beta}(\frac {%
\partial\varphi}{\partial g_{\alpha\beta}(x)}+\frac{\partial\varphi}{%
\partial g_{\beta\alpha}(x)})\delta g_{\alpha\beta}(x)
\end{align*}%
\begin{align*}
& +\sum\nolimits_{\alpha}\frac{\partial\varphi}{\partial\partial_{\mu
}g_{\alpha\alpha}(x)}\delta\partial_{\mu}g_{\alpha\alpha}(x)+\sum
\nolimits_{\alpha<\beta}(\frac{\partial\varphi}{\partial\partial_{\mu
}g_{\alpha\beta}(x)} \\
& +\frac{\partial\varphi}{\partial\partial_{\mu}g_{\beta\alpha}(x)}%
)\delta\partial_{\mu}g_{\alpha\beta}(x)+\sum\nolimits_{\alpha,\mu}\frac{%
\partial\varphi}{\partial\partial_{\mu}\partial_{\mu}g_{\alpha\alpha }(x)}%
\delta\partial_{\mu}\partial_{\mu}g_{\alpha\alpha}(x)
\end{align*}%
\begin{align*}
& +\sum\nolimits_{\alpha<\beta,\mu}(\frac{\partial\varphi}{\partial
\partial_{\mu}\partial_{\mu}g_{\alpha\beta}(x)}+\frac{\partial\varphi }{%
\partial\partial_{\mu}\partial_{\mu}g_{\beta\alpha}(x)})\delta\partial_{\mu
}\partial_{\mu}g_{\alpha\beta}(x) \\
& +\sum\nolimits_{\alpha,\mu<\nu}(\frac{\partial\varphi}{\partial
\partial_{\mu}\partial_{\nu}g_{\alpha\alpha}(x)}+\frac{\partial\varphi }{%
\partial\partial_{\nu}\partial_{\mu}g_{\alpha\alpha}(x)})\delta\partial
_{\mu}\partial_{\nu}g_{\alpha\alpha}(x)
\end{align*}%
\begin{align}
& +\sum\nolimits_{\alpha<\beta,\mu<\nu}(\frac{\partial\varphi}{\partial
\partial_{\mu}\partial_{\nu}g_{\alpha\beta}(x)}+\frac{\partial\varphi }{%
\partial\partial_{\nu}\partial_{\mu}g_{\alpha\beta}(x)}+  \notag \\
& \frac{\partial\varphi}{\partial\partial_{\mu}\partial_{\nu}g_{\beta\alpha
}(x)}+\frac{\partial\varphi}{\partial\partial_{\nu}\partial_{\mu}g_{\beta%
\alpha}(x)})\delta\partial_{\mu}\partial_{\nu}g_{\alpha\beta}(x)]|_{D} 
\tag{B7}
\end{align}
Because all the variations on the RHS of (B7) are independent, we get 
\begin{equation*}
(\frac{\partial\varphi}{\partial g_{\alpha\beta}(x)}+\frac{\partial\varphi }{%
\partial g_{\beta\alpha}(x)})|_{D}=(\frac{\partial\psi}{\partial
g_{\alpha\beta}(x)}+\frac{\partial\psi}{\partial g_{\beta\alpha}(x)})|_{D}
\end{equation*}%
\begin{equation*}
(\frac{\partial\varphi}{\partial\partial_{\mu}g_{\alpha\beta}(x)}+\frac{%
\partial\varphi}{\partial\partial_{\mu}g_{\beta\alpha}(x)})|_{D}=(\frac{%
\partial\psi}{\partial\partial_{\mu}g_{\alpha\beta}(x)}+\frac {\partial\psi}{%
\partial\partial_{\mu}g_{\beta\alpha}(x)})|_{D}
\end{equation*}

\begin{equation*}
(\frac{\partial \varphi }{\partial \partial _{\mu }\partial _{\nu }g_{\alpha
\beta }(x)}+\frac{\partial \varphi }{\partial \partial _{\nu }\partial _{\mu
}g_{\alpha \beta }(x)}+\frac{\partial \psi }{\partial \partial _{\mu
}\partial _{\nu }g_{\beta \alpha }(x)}+\frac{\partial \psi }{\partial
\partial _{\nu }\partial _{\mu }g_{\beta \alpha }(x)})|_{D}
\end{equation*}%
\begin{equation}
=(\frac{\partial \psi }{\partial \partial _{\mu }\partial _{\nu }g_{\alpha
\beta }(x)}+\frac{\partial \psi }{\partial \partial _{\nu }\partial _{\mu
}g_{\alpha \beta }(x)}+\frac{\partial \psi }{\partial \partial _{\mu
}\partial _{\nu }g_{\beta \alpha }(x)}+\frac{\partial \psi }{\partial
\partial _{\nu }\partial _{\mu }g_{\beta \alpha }(x)})|_{D}  \tag{B8}
\end{equation}%
This tells us, say,%
\begin{equation}
\frac{\partial R}{\partial g_{\alpha \beta }(x)}|_{D}=\frac{1}{2}(\frac{%
\partial \varphi }{\partial g_{\alpha \beta }(x)}+\frac{\partial \varphi }{%
\partial g_{\beta \alpha }(x)})|_{D}  \tag{B9}
\end{equation}%
\begin{equation}
\frac{\partial R}{\partial \partial _{\mu }g_{\alpha \beta }(x)}|_{D}=\frac{1%
}{2}(\frac{\partial \varphi }{\partial \partial _{\mu }g_{\alpha \beta }(x)}+%
\frac{\partial \varphi }{\partial \partial _{\mu }g_{\beta \alpha }(x)})|_{D}
\tag{B10}
\end{equation}%
\begin{align}
\frac{\partial R}{\partial \partial _{\mu }\partial _{\nu }g_{\alpha \beta
}(x)}|_{D}& =\frac{1}{4}(\frac{\partial \varphi }{\partial \partial _{\mu
}\partial _{\nu }g_{\alpha \beta }(x)}+\frac{\partial \varphi }{\partial
\partial _{\mu }\partial _{\nu }g_{\beta \alpha }(x)}  \notag \\
& +\frac{\partial \varphi }{\partial \partial _{\nu }\partial _{\mu
}g_{\alpha \beta }(x)}+\frac{\partial \varphi }{\partial \partial _{\nu
}\partial _{\mu }g_{\beta \alpha }(x)})|_{D}  \tag{B11}
\end{align}%
where $R(g(x),\partial g(x),\partial ^{2}g(x))$ is the \textquotedblleft
standard\textquotedblright\ expression for $R$, and $\varphi (g(x)$, $%
\partial g(x)$, $\partial ^{2}g(x))$ is any expression from (B1). Eqns.(B9),
(B10) and (B11) tell us,%
\begin{align*}
\frac{\partial R}{\partial g_{\alpha \beta }(x)}\delta g_{\alpha \beta
}(x)|_{D}& =\frac{\partial \varphi }{\partial g_{\alpha \beta }(x)}\delta
g_{\alpha \beta }(x)|_{D}, \\
\frac{\partial R}{\partial \partial _{\mu }g_{\alpha \beta }(x)}\delta
\partial _{\rho }g_{\alpha \beta }(x)|_{D}& =\frac{\partial \varphi }{%
\partial \partial _{\mu }g_{\alpha \beta }(x)}\delta \partial _{\rho
}g_{\alpha \beta }(x)|_{D},
\end{align*}%
\begin{equation}
\frac{\partial R}{\partial \partial _{\mu }\partial _{\nu }g_{\alpha \beta
}(x)}\delta \partial _{\rho }\partial _{\sigma }g_{\alpha \beta }(x)|_{D}=%
\frac{\partial \varphi }{\partial \partial _{\mu }\partial _{\nu }g_{\alpha
\beta }(x)}\delta \partial _{\rho }\partial _{\sigma }g_{\alpha \beta
}(x)|_{D},\ldots  \tag{B12}
\end{equation}

\section{Appendix}

\begin{proposition}
\begin{align}
\frac{\partial}{\partial x^{\kappa}}\left\vert \frac{\partial y}{\partial x}%
\right\vert & =\left\vert \frac{\partial y}{\partial x}\right\vert \frac{%
\partial}{\partial y^{\lambda}}\left( \frac{\partial y^{\lambda}}{\partial
x^{\kappa}}\right)  \tag{C1} \\
\frac{\partial}{\partial x^{\lambda}}[\left\vert \frac{\partial y}{\partial x%
}\right\vert \frac{\partial x^{\lambda}}{\partial y^{\mu}}] & =0
\end{align}
\end{proposition}

\begin{proof}
\begin{align}
\frac{\partial}{\partial x^{\kappa}}\left\vert \frac{\partial y}{\partial x}%
\right\vert & =\frac{\partial}{\partial x^{\kappa}}(\frac{\partial
y^{\lambda}}{\partial x^{\alpha}})\frac{\partial}{\partial(\frac{\partial
y^{\lambda}}{\partial x^{\alpha}})}\left\vert \frac{\partial y}{\partial x}%
\right\vert  \notag \\
& =\frac{\partial^{2}y^{\lambda}}{\partial x^{\kappa}\partial x^{\alpha}}%
\frac{\partial x^{\alpha}}{\partial y^{\lambda}}\left\vert \frac{\partial y}{%
\partial x}\right\vert  \notag \\
& =\frac{\partial}{\partial y^{\lambda}}(\frac{\partial y^{\lambda}}{%
\partial x^{\kappa}})\left\vert \frac{\partial y}{\partial x}\right\vert 
\tag{C2}
\end{align}
\end{proof}

\section{Appendix}

\begin{proposition}
Let%
\begin{align}
I_{x}^{\kappa} & =:(\frac{\partial R}{\partial\partial_{\kappa}g_{\alpha%
\beta}(x)}-\partial_{\mu}\frac{\partial R}{\partial\partial_{\kappa
}\partial_{\mu}g_{\alpha\beta}(x)}-\Gamma_{\nu\mu}^{\nu}(x)\frac{\partial R}{%
\partial\partial_{\kappa}\partial_{\mu}g_{\alpha\beta}(x)})\times  \notag \\
& (\delta g_{\alpha\beta}(x)-\delta x^{\rho}\partial_{\rho}g_{\alpha\beta
}(x))+\frac{\partial R}{\partial\partial_{\kappa}\partial_{\mu}g_{\alpha%
\beta }(x)}\partial_{\mu}(\delta g_{\alpha\beta}(x)-\delta
x^{\rho}\partial_{\rho }g_{\alpha\beta}(x)).  \tag{D1}
\end{align}
Then%
\begin{equation}
I_{y}^{\lambda}=\frac{\partial y^{\lambda}}{\partial x^{\kappa}}%
I_{x}^{\kappa }.  \tag{D2}
\end{equation}
\end{proposition}

\begin{proof}
\begin{align}
R& =g^{\alpha \beta }g^{\rho \sigma }(\partial _{\alpha }\partial _{\rho
}g_{\beta \sigma }-\partial _{\alpha }\partial _{\beta }g_{\rho \sigma
})+g^{\alpha \beta }g^{\rho \sigma }g^{\xi \eta }(\partial _{\alpha
}g_{\beta \rho }\partial _{\sigma }g_{\xi \eta }+  \notag \\
& \frac{3}{4}\partial _{\alpha }g_{\rho \xi }\partial _{\beta }g_{\sigma
\eta }-\frac{1}{4}\partial _{\xi }g_{\alpha \beta }\partial _{\eta }g_{\rho
\sigma }-\frac{1}{2}\partial _{\rho }g_{\alpha \xi }\partial _{\eta
}g_{\beta \sigma }-\partial _{\alpha }g_{\beta \rho }\partial _{\xi
}g_{\sigma \eta })  \tag{D3}
\end{align}%
\begin{equation}
\frac{\partial R}{\partial \partial _{\kappa }\partial _{\mu }g_{\alpha
\beta }}=\frac{1}{2}(g^{\alpha \kappa }g^{\beta \mu }+g^{\alpha \mu
}g^{\beta \kappa })-g^{\alpha \beta }g^{\kappa \mu }  \tag{D4}
\end{equation}%
Note that $\frac{\partial R}{\partial \partial _{\kappa }\partial _{\mu
}g_{\alpha \beta }}$ is a $(4,0)$-tensor, symmetrical for $(\kappa ,\mu )$,
and for $(\alpha ,\beta )$.%
\begin{align}
\frac{\partial R}{\partial \partial _{\kappa }g_{\alpha \beta }}& =\partial
_{\mu }g_{\xi \eta }[g^{\alpha \kappa }g^{\beta \mu }g^{\xi \eta }+g^{\alpha
\beta }g^{\kappa \xi }g^{\mu \eta }+\frac{3}{2}g^{\alpha \xi }g^{\beta \eta
}g^{\kappa \mu }  \notag \\
& -\frac{1}{2}g^{\alpha \beta }g^{\kappa \mu }g^{\xi \eta }-g^{\alpha \xi
}g^{\beta \mu }g^{\kappa \eta }-g^{\alpha \kappa }g^{\beta \xi }g^{\mu \eta
}-g^{\alpha \xi }g^{\beta \kappa }g^{\mu \eta }]  \notag \\
& =:\partial _{\mu }g_{\xi \eta }B^{\kappa \alpha \beta \mu \xi \eta } 
\tag{D5}
\end{align}%
\begin{align*}
\partial _{\mu }\frac{\partial R}{\partial \partial _{\kappa }\partial _{\mu
}g_{\alpha \beta }}& =\frac{1}{2}(\partial _{\mu }g^{\alpha \kappa }g^{\beta
\mu }+\partial _{\mu }g^{\alpha \mu }g^{\beta \kappa })-\partial _{\mu
}g^{\alpha \beta }g^{\kappa \mu } \\
& +\frac{1}{2}(g^{\alpha \kappa }\partial _{\mu }g^{\beta \mu }+g^{\alpha
\mu }\partial _{\mu }g^{\beta \kappa })-g^{\alpha \beta }\partial _{\mu
}g^{\kappa \mu }
\end{align*}%
\begin{align}
& =\partial _{\mu }g_{\xi \eta }[-\frac{1}{2}g^{\alpha \xi }g^{\eta \kappa
}g^{\beta \mu }-\frac{1}{2}g^{\alpha \xi }g^{\eta \mu }g^{\beta \kappa
}+g^{\alpha \xi }g^{\eta \beta }g^{\kappa \mu }]  \notag \\
& +\partial _{\mu }g_{\xi \eta }[-\frac{1}{2}g^{\alpha \kappa }g^{\beta \xi
}g^{\eta \mu }-\frac{1}{2}g^{\alpha \mu }g^{\beta \xi }g^{\kappa \mu
}+g^{\alpha \beta }g^{\kappa \xi }g^{\eta \mu }]  \notag \\
& =:\partial _{\mu }g_{\xi \eta }A^{\kappa \alpha \beta \mu \xi \eta } 
\tag{D6}
\end{align}%
Note that $A^{\kappa \alpha \beta \mu \xi \eta }$ and $B^{\kappa \alpha
\beta \mu \xi \eta }$ are $(6,0)$-tensors.Let%
\begin{equation*}
C^{\kappa \alpha \beta \mu \xi \eta }(x)=:B^{\kappa \alpha \beta \mu \xi
\eta }(x)-A^{\kappa \alpha \beta \mu \xi \eta }(x)
\end{equation*}%
\begin{align}
& =-\frac{1}{2}g^{\alpha \xi }(x)g^{\beta \mu }(x)g^{\kappa \eta }(x)+\frac{1%
}{2}g^{\alpha \mu }(x)g^{\beta \xi }(x)g^{\kappa \eta }(x)  \notag \\
& -\frac{1}{2}g^{\alpha \xi }(x)g^{\beta \kappa }(x)g^{\mu \eta }(x)-\frac{1%
}{2}g^{\alpha \kappa }(x)g^{\beta \xi }(x)g^{\mu \eta }(x)  \notag
\end{align}%
\begin{equation}
+\frac{1}{2}g^{\alpha \xi }(x)g^{\beta \eta }(x)g^{\kappa \mu }(x)+g^{\alpha
\kappa }(x)g^{\beta \mu }(x)g^{\xi \eta }(x)-\frac{1}{2}g^{\alpha \beta
}(x)g^{\kappa \mu }(x)g^{\xi \eta }(x)  \tag{D7}
\end{equation}%
\begin{equation}
\Gamma _{\alpha \nu }^{\nu }(x)=\frac{\partial y^{\beta }}{\partial
x^{\alpha }}\Gamma _{\beta \gamma }^{\gamma }(y)+\frac{\partial }{\partial
y^{\sigma }}(\frac{\partial y^{\sigma }}{\partial x^{\alpha }})  \tag{D8}
\end{equation}%
Then%
\begin{align*}
I_{x}^{\kappa }& =[C^{\kappa \alpha \beta \mu \xi \eta }(x)\partial _{\mu
}g_{\xi \eta }(x)-\Gamma _{\mu \nu }^{\nu }(x)\frac{\partial R}{\partial
\partial _{\kappa }\partial _{\mu }g_{\alpha \beta }}]\overline{\delta }%
g_{\alpha \beta }(x) \\
& +\frac{\partial R}{\partial \partial _{\kappa }\partial _{\mu }g_{\alpha
\beta }(x)}\partial _{\mu }\overline{\delta }g_{\alpha \beta }(x)
\end{align*}%
\begin{align*}
& =\frac{\partial x^{\kappa }}{\partial y^{\kappa ^{\prime }}}[C^{\kappa
^{\prime }\alpha ^{\prime }\beta ^{\prime }\mu ^{\prime }\xi ^{\prime }\eta
^{\prime }}(y)\frac{\partial x^{\xi }}{\partial y^{\xi ^{\prime }}}\frac{%
\partial x^{\eta }}{\partial y^{\eta ^{\prime }}}\frac{\partial }{\partial
y^{\mu ^{\prime }}}(\frac{\partial y^{\xi "}}{\partial x^{\xi }}\frac{%
\partial y^{\eta "}}{\partial x^{\eta }}g_{\xi "\eta "}(y)) \\
& -\Gamma _{\nu ^{\prime }\mu \prime }^{\nu ^{\prime }}(y)\frac{\partial R}{%
\partial \partial _{\kappa ^{\prime }}\partial _{\mu ^{\prime }}g_{\alpha
^{\prime }\beta ^{\prime }}(y)}+\frac{\partial }{\partial x^{\mu }}(\frac{%
\partial x^{\mu }}{\partial y^{\mu ^{\prime }}})\frac{\partial R}{\partial
\partial _{\kappa ^{\prime }}\partial _{\mu ^{\prime }}g_{\alpha ^{\prime
}\beta ^{\prime }}(y)}]\overline{\delta }g_{\alpha ^{\prime }\beta ^{\prime
}}(y)
\end{align*}%
\begin{equation*}
+\frac{\partial x^{\kappa }}{\partial y^{\kappa ^{\prime }}}\frac{\partial R%
}{\partial \partial _{\kappa ^{\prime }}\partial _{\mu ^{\prime }}g_{\alpha
^{\prime }\beta ^{\prime }}(y)}\frac{\partial x^{\alpha }}{\partial
y^{\alpha ^{\prime }}}\frac{\partial x^{\beta }}{\partial y^{\beta ^{\prime
}}}\frac{\partial }{\partial y^{\mu ^{\prime }}}[\frac{\partial x^{\alpha }}{%
\partial y^{\alpha "}}\frac{\partial x^{\beta }}{\partial y^{\beta "}}%
\overline{\delta }g_{\alpha "\beta "}(y)]
\end{equation*}%
\begin{align}
& =\frac{\partial x^{\kappa }}{\partial y^{\kappa ^{\prime }}}\{[C^{\kappa
^{\prime }\alpha ^{\prime }\beta ^{\prime }\mu ^{\prime }\xi ^{\prime }\eta
^{\prime }}(y)\partial _{\mu ^{\prime }}g_{\xi ^{\prime }\eta ^{\prime
}}(y)-\Gamma _{\nu ^{\prime }\mu \prime }^{\nu ^{\prime }}(y)\frac{\partial R%
}{\partial \partial _{\kappa ^{\prime }}\partial _{\mu ^{\prime }}g_{\alpha
^{\prime }\beta ^{\prime }}(y)}]\times  \notag \\
& \overline{\delta }g_{\alpha ^{\prime }\beta ^{\prime }}(y)+\frac{\partial R%
}{\partial \partial _{\kappa ^{\prime }}\partial _{\mu ^{\prime }}g_{\alpha
^{\prime }\beta ^{\prime }}(y)}\partial _{\mu ^{\prime }}\overline{\delta }%
g_{\alpha ^{\prime }\beta \prime }(y)\}+\frac{\partial x^{\kappa }}{\partial
y^{\kappa ^{\prime }}}rest  \notag
\end{align}%
\begin{equation*}
=\frac{\partial x^{\kappa }}{\partial y^{\kappa ^{\prime }}}I_{y}^{\kappa
^{\prime }}+\frac{\partial x^{\kappa }}{\partial y^{\kappa ^{\prime }}}rest
\end{equation*}%
where%
\begin{equation*}
rest=:[-C^{\kappa ^{\prime }\alpha ^{\prime }\beta ^{\prime }\mu ^{\prime
}\xi ^{\prime }\eta ^{\prime }}(y)\frac{\partial ^{2}x^{\xi }}{\partial
y^{\mu ^{\prime }}\partial y^{\xi ^{\prime }}}\frac{\partial y^{\xi "}}{%
\partial x^{\xi }}g_{\xi "\eta \prime }(y)
\end{equation*}%
\begin{align}
& -C^{\kappa ^{\prime }\alpha ^{\prime }\beta ^{\prime }\mu ^{\prime }\xi
^{\prime }\eta ^{\prime }}(y)\frac{\partial ^{2}x^{\eta }}{\partial y^{\mu
^{\prime }}\partial y^{\eta ^{\prime }}}\frac{\partial y^{\eta "}}{\partial
x^{\eta }}g_{\xi ^{\prime }\eta "}(y)  \notag \\
& +\frac{\partial }{\partial x^{\mu }}(\frac{\partial x^{\mu }}{\partial
y^{\mu ^{\prime }}})\frac{\partial R}{\partial \partial _{\kappa ^{\prime
}}\partial _{\mu ^{\prime }}g_{\alpha ^{\prime }\beta ^{\prime }}(y)}]%
\overline{\delta }g_{\alpha ^{\prime }\beta ^{\prime }}(y)  \tag{D9}
\end{align}%
We are going to show that $rest$ vanishes. Its first term is%
\begin{align*}
& -C^{\kappa ^{\prime }\alpha ^{\prime }\beta ^{\prime }\mu ^{\prime }\xi
^{\prime }\eta ^{\prime }}(y)\frac{\partial ^{2}x^{\xi }}{\partial y^{\mu
^{\prime }}\partial y^{\xi ^{\prime }}}\frac{\partial y^{\xi "}}{\partial
x^{\xi }}g_{\xi "\eta \prime }(y)\overline{\delta }g_{\alpha ^{\prime }\beta
^{\prime }}(y) \\
& =[g^{\alpha ^{\prime }\xi ^{\prime }}(y)g^{\beta ^{\prime }\kappa ^{\prime
}}(y)g^{\mu ^{\prime }\eta ^{\prime }}(y)-\frac{1}{2}g^{\alpha ^{\prime }\xi
^{\prime }}(y)g^{\beta ^{\prime }\eta ^{\prime }}(y)g^{\kappa ^{\prime }\mu
^{\prime }}(y)
\end{align*}%
\begin{align*}
& -g^{\alpha ^{\prime }\kappa \prime }(y)g^{\beta ^{\prime }\mu ^{\prime
}}(y)g^{\xi ^{\prime }\eta ^{\prime }}(y)+\frac{1}{2}g^{\alpha ^{\prime
}\beta ^{\prime }}(y)g^{\kappa ^{\prime }\mu ^{\prime }}(y)g^{\xi ^{\prime
}\eta ^{\prime }}(y)]\times \\
& \frac{\partial ^{2}x^{\xi }}{\partial y^{\mu ^{\prime }}\partial y^{\xi
^{\prime }}}\frac{\partial y^{\xi "}}{\partial x^{\xi }}g_{\xi "\eta \prime
}(y)\overline{\delta }g_{\alpha ^{\prime }\beta ^{\prime }}(y)
\end{align*}%
\begin{align*}
& =g^{\alpha ^{\prime }\xi ^{\prime }}(y)g^{\beta ^{\prime }\kappa ^{\prime
}}(y)\frac{\partial }{\partial x^{\xi }}(\frac{\partial x^{\xi }}{\partial
y^{\xi ^{\prime }}})\overline{\delta }g_{\alpha ^{\prime }\beta ^{\prime
}}(y)\text{ \ \ \ \ (A)} \\
& -\frac{1}{2}g^{\alpha ^{\prime }\xi ^{\prime }}(y)g^{\kappa ^{\prime }\mu
^{\prime }}(y)\frac{\partial ^{2}x^{\xi }}{\partial y^{\mu ^{\prime
}}\partial y^{\xi ^{\prime }}}\frac{\partial y^{\beta ^{\prime }}}{\partial
x^{\xi }}\overline{\delta }g_{\alpha ^{\prime }\beta ^{\prime }}(y)\text{ \
\ (B)}
\end{align*}%
\begin{align}
& -g^{\alpha ^{\prime }\kappa \prime }(y)g^{\beta ^{\prime }\mu ^{\prime
}}(y)\frac{\partial }{\partial x^{\xi }}(\frac{\partial x^{\xi }}{\partial
y^{\mu ^{\prime }}})\overline{\delta }g_{\alpha ^{\prime }\beta ^{\prime
}}(y)\text{ \ (A)}  \notag \\
& +\frac{1}{2}g^{\alpha ^{\prime }\beta ^{\prime }}(y)g^{\kappa \prime \mu
\prime }(y)\frac{\partial }{\partial x^{\xi }}(\frac{\partial x^{\xi }}{%
\partial y^{\mu \prime }})\overline{\delta }g_{\alpha ^{\prime }\beta
^{\prime }}(y)\text{ \ \ ((C)}  \tag{D10}
\end{align}%
The second term is%
\begin{align*}
& -C^{\kappa ^{\prime }\alpha ^{\prime }\beta ^{\prime }\mu ^{\prime }\xi
^{\prime }\eta ^{\prime }}(y)\frac{\partial ^{2}x^{\eta }}{\partial y^{\mu
^{\prime }}\partial y^{\eta ^{\prime }}}\frac{\partial y^{\eta "}}{\partial
x^{\eta }}g_{\xi ^{\prime }\eta "}(y)\overline{\delta }g_{\alpha ^{\prime
}\beta ^{\prime }}(y) \\
& =[g^{\alpha ^{\prime }\xi ^{\prime }}(y)g^{\beta ^{\prime }\kappa ^{\prime
}}(y)g^{\mu ^{\prime }\eta ^{\prime }}(y)-\frac{1}{2}g^{\alpha ^{\prime }\xi
^{\prime }}(y)g^{\beta ^{\prime }\eta ^{\prime }}(y)g^{\kappa ^{\prime }\mu
^{\prime }}(y)
\end{align*}%
\begin{align*}
& -g^{\alpha ^{\prime }\kappa \prime }(y)g^{\beta ^{\prime }\mu ^{\prime
}}(y)g^{\xi ^{\prime }\eta ^{\prime }}(y)+\frac{1}{2}g^{\alpha ^{\prime
}\beta ^{\prime }}(y)g^{\kappa ^{\prime }\mu ^{\prime }}(y)g^{\xi ^{\prime
}\eta ^{\prime }}(y)]\times \\
& \frac{\partial ^{2}x^{\eta }}{\partial y^{\mu ^{\prime }}\partial y^{\eta
^{\prime }}}\frac{\partial y^{\eta "}}{\partial x^{\eta }}g_{\xi ^{\prime
}\eta "}(y)\overline{\delta }g_{\alpha ^{\prime }\beta ^{\prime }}(y)
\end{align*}%
\begin{align*}
& =g^{\beta ^{\prime }\kappa ^{\prime }}(y)g^{\mu ^{\prime }\eta ^{\prime
}}(y)\frac{\partial ^{2}x^{\eta }}{\partial y^{\mu ^{\prime }}\partial
y^{\eta ^{\prime }}}\frac{\partial y^{\alpha ^{\prime }}}{\partial x^{\eta }}%
\overline{\delta }g_{\alpha ^{\prime }\beta ^{\prime }}(y)\text{ \ \ (D)} \\
& -\frac{1}{2}g^{\beta ^{\prime }\eta ^{\prime }}(y)g^{\kappa ^{\prime }\mu
^{\prime }}(y)\frac{\partial ^{2}x^{\eta }}{\partial y^{\mu ^{\prime
}}\partial y^{\eta ^{\prime }}}\frac{\partial y^{\alpha ^{\prime }}}{%
\partial x^{\eta }}\overline{\delta }g_{\alpha ^{\prime }\beta ^{\prime }}(y)%
\text{ \ (B)}
\end{align*}%
\begin{align}
& -g^{\alpha ^{\prime }\kappa \prime }(y)g^{\beta ^{\prime }\mu ^{\prime
}}(y)\frac{\partial }{\partial x^{\eta }}(\frac{\partial x^{\eta }}{\partial
y^{\mu ^{\prime }}})\overline{\delta }g_{\alpha ^{\prime }\beta ^{\prime
}}(y)\text{ \ \ (A)}  \notag \\
& +\frac{1}{2}g^{\alpha ^{\prime }\beta ^{\prime }}(y)g^{\kappa ^{\prime
}\mu ^{\prime }}(y)\frac{\partial }{\partial x^{\eta }}(\frac{\partial
x^{\eta }}{\partial y^{\mu ^{\prime }}})\overline{\delta }g_{\alpha ^{\prime
}\beta ^{\prime }}(y)\text{ \ \ (C)}  \tag{D11}
\end{align}%
The third term is%
\begin{align*}
& \frac{\partial }{\partial x^{\mu }}(\frac{\partial x^{\mu }}{\partial
y^{\mu ^{\prime }}})\frac{\partial R}{\partial \partial _{\kappa ^{\prime
}}\partial _{\mu ^{\prime }}g_{\alpha ^{\prime }\beta ^{\prime }}(y)}%
\overline{\delta }g_{\alpha ^{\prime }\beta ^{\prime }}(y) \\
& =\frac{\partial }{\partial x^{\mu }}(\frac{\partial x^{\mu }}{\partial
y^{\mu ^{\prime }}})(g^{\alpha ^{\prime }\kappa ^{\prime }}g^{\beta ^{\prime
}\mu ^{\prime }}-g^{\alpha ^{\prime }\beta \prime }g^{\kappa ^{\prime }\mu
^{\prime }})\overline{\delta }g_{\alpha ^{\prime }\beta ^{\prime }}(y)
\end{align*}%
\begin{align}
& =\frac{\partial }{\partial x^{\mu }}(\frac{\partial x^{\mu }}{\partial
y^{\mu ^{\prime }}})g^{\alpha ^{\prime }\kappa ^{\prime }}g^{\beta ^{\prime
}\mu ^{\prime }}\overline{\delta }g_{\alpha ^{\prime }\beta ^{\prime }}(y)%
\text{ \ \ (A)}  \notag \\
& -\frac{\partial }{\partial x^{\mu }}(\frac{\partial x^{\mu }}{\partial
y^{\mu ^{\prime }}})g^{\alpha ^{\prime }\beta \prime }g^{\kappa ^{\prime
}\mu ^{\prime }}\overline{\delta }g_{\alpha ^{\prime }\beta ^{\prime }}(y)%
\text{ \ \ (C)}  \tag{D12}
\end{align}%
The forth term is%
\begin{align*}
& -\frac{\partial R}{\partial \partial _{\kappa ^{\prime }}\partial _{\mu
^{\prime }}g_{\alpha "\beta ^{\prime }}(y)}\frac{\partial ^{2}x^{\alpha }}{%
\partial y^{\mu ^{\prime }}\partial y^{\alpha "}}\frac{\partial y^{\alpha
\prime }}{\partial x^{\alpha }}\overline{\delta }g_{\alpha \prime \beta
\prime }(y) \\
& =-\frac{\partial ^{2}x^{\alpha }}{\partial y^{\mu ^{\prime }}\partial
y^{\alpha "}}\frac{\partial y^{\alpha \prime }}{\partial x^{\alpha }}%
(g^{\alpha "\kappa ^{\prime }}g^{\beta ^{\prime }\mu ^{\prime }}-g^{\alpha
"\beta \prime }g^{\kappa ^{\prime }\mu ^{\prime }})\overline{\delta }%
g_{\alpha ^{\prime }\beta ^{\prime }}(y)
\end{align*}%
\begin{align}
& =-\frac{\partial ^{2}x^{\alpha }}{\partial y^{\mu ^{\prime }}\partial
y^{\alpha "}}\frac{\partial y^{\alpha \prime }}{\partial x^{\alpha }}%
g^{\alpha "\kappa ^{\prime }}g^{\beta ^{\prime }\mu ^{\prime }}\overline{%
\delta }g_{\alpha ^{\prime }\beta ^{\prime }}(y)\text{ \ \ (B)}  \notag \\
& +\frac{\partial ^{2}x^{\alpha }}{\partial y^{\mu ^{\prime }}\partial
y^{\alpha "}}\frac{\partial y^{\alpha \prime }}{\partial x^{\alpha }}%
g^{\alpha "\beta \prime }g^{\kappa ^{\prime }\mu ^{\prime }})\overline{%
\delta }g_{\alpha ^{\prime }\beta ^{\prime }}(y)\text{ \ \ (B)}  \tag{D13}
\end{align}%
The last term is%
\begin{align*}
& -\frac{\partial R}{\partial \partial _{\kappa ^{\prime }}\partial _{\mu
^{\prime }}g_{\alpha ^{\prime }\beta "}(y)}\frac{\partial ^{2}x^{\beta }}{%
\partial y^{\mu ^{\prime }}\partial y^{\beta "}}\frac{\partial y^{\beta
\prime }}{\partial x^{\beta }}\overline{\delta }g_{\alpha \prime \beta
\prime }(y) \\
& =-\frac{\partial ^{2}x^{\beta }}{\partial y^{\mu ^{\prime }}\partial
y^{\beta "}}\frac{\partial y^{\beta \prime }}{\partial x^{\beta }}(g^{\alpha
\prime \kappa ^{\prime }}g^{\beta "\mu ^{\prime }}-g^{\alpha \prime \beta
"}g^{\kappa ^{\prime }\mu ^{\prime }})\overline{\delta }g_{\alpha ^{\prime
}\beta ^{\prime }}(y)
\end{align*}%
\begin{align}
& =-\frac{\partial ^{2}x^{\beta }}{\partial y^{\mu ^{\prime }}\partial
y^{\beta "}}\frac{\partial y^{\beta \prime }}{\partial x^{\beta }}g^{\alpha
\prime \kappa ^{\prime }}g^{\beta "\mu ^{\prime }}\overline{\delta }%
g_{\alpha ^{\prime }\beta ^{\prime }}(y)\text{ \ \ (D)}  \notag \\
& +\frac{\partial ^{2}x^{\beta }}{\partial y^{\mu ^{\prime }}\partial
y^{\beta "}}\frac{\partial y^{\beta \prime }}{\partial x^{\beta }}g^{\alpha
\prime \beta "}g^{\kappa ^{\prime }\mu ^{\prime }}\overline{\delta }%
g_{\alpha ^{\prime }\beta ^{\prime }}(y)\text{ \ \ (B)}  \tag{D14}
\end{align}%
All the terms marked (A) cancel each other, all the terms marked (B) cancel
each other, etc. Therfore we get%
\begin{equation}
I_{x}^{\kappa }=\frac{\partial x^{\kappa }}{\partial y^{k^{\prime }}}%
I_{y}^{\kappa ^{\prime }}  \tag{D15}
\end{equation}%
That is (D2).
\end{proof}

\end{document}